\documentclass{amsart}
\usepackage{amssymb}
\usepackage{verbatim}
\usepackage{mathrsfs}
\usepackage{dsfont}

\newtheorem{theorem}{Theorem}[section]
\newtheorem{lemma}[theorem]{Lemma}
\newtheorem{corollary}[theorem]{Corollary}
\newtheorem{proposition}[theorem]{Proposition}

\theoremstyle{definition}
\newtheorem{definition}[theorem]{Definition}

\theoremstyle{remark}
\newtheorem{remark}[theorem]{Remark}

\newcommand{\I}{{\mathds{1}}}

\newcommand{\cH}{{\mathcal H}}

\newcommand{\cM}{{\mathcal M}}
\newcommand{\qM}{{\mathfrak{M}}}

\newcommand{\bR}{{\mathbb{R}}}

\newcommand{\Rn}{\mathbb{R}} 
\newcommand{\Nn}{\mathbb{N}} 
\newcommand{\Cn}{\mathbb{C}} 

\numberwithin{equation}{section}

\newcommand{\Tr}{\mathrm{Tr}}

\newcommand{\jed}{{\mathbb{I}}}
\begin{document}

\title[]{On Entropy for general quantum systems}

\author{W. A. Majewski}

\address{Institute of Theoretical Physics and Astrophysics, The Gdansk University, Wita Stwosza 57,\\
Gdansk, 80-952, Poland and Unit for BMI, North-West-University, Potchefstroom, South Africa}
\email{fizwam@univ.gda.pl}

\author{L. E. Labuschagne}

\address{DST-NRF CoE in Math. and Stat. Sci,\\ Unit for BMI,\\ Internal Box 209, School of Comp., Stat., $\&$ Math. Sci.\\
NWU, PVT. BAG X6001, 2520 Potchefstroom\\ South Africa}
\email{Louis.Labuschagne@nwu.ac.za}

\begin{abstract}
In these notes we will give an overview and road map for a definition and characterization of (relative) entropy for both classical and quantum systems. In other words, we will provide a consistent treatment of entropy which can be applied within the recently developed Orlicz space based approach to large systems. This means that the proposed approach successfully provides a refined framework for the treatment of entropy in each of classical statistical physics, Dirac's formalism of Quantum Mechanics, large systems of quantum statistical physics, and finally also for Quantum Field Theory.
\end{abstract}
\date{\today}
\thanks{We are very grateful to Fumio Hiai for his comments on the first version of the paper and for bringing Ohya-Petz result on 
the relative entropy to our attention. The contribution of L. E. Labuschagne is based on research partially supported by the National 
Research Foundation (IPRR Grant 96128). Any opinion, findings and conclusions or recommendations expressed in this material, are those 
of the author, and therefore the NRF do not accept any liability in regard thereto.}

\maketitle

Despite the efforts of many authors over a very long period of time, gaining a deeper understanding of entropy remains one of the most 
important and intriguing challenges in the physics of large systems - a challenge still receiving the close attention of many prominent 
authors. See for example \cite{Wit}. In this endeavour the techniques available for the quantum framework still lack the refinement of 
those available for classical systems. On this point, Dirac's formalism for Quantum Mechanics and von Neumann's definition of entropy in 
the context of $B(H)$, does however provide a ``template'' for developing techniques for the description and study of entropy in the 
context of tracial von Neumann algebras. One possible way in which von Neumann's ideas could be refined to provide a ``good'' 
description of states with well-defined entropy in the tracial case, was fully described in \cite{[2]}. As is shown in that paper, a 
successful description of states with entropy can be achieved on condition that the more common framework for quantum theory based on 
the pair of spaces $\langle L^\infty, L^1\rangle$, is replaced with a formalism based on the pair of Orlicz spaces 
$\langle L^{\cosh-1}, L\log(L+1)\rangle$. An important point worth noting (also pointed out in \cite{[2]}), is that this axiomatic shift 
leaves Dirac's formalism intact! However not all quantum systems correspond to tracial von Neumann algebras. (Note for example that the 
local algebras of Quantum Field Theory are type $III_1$.) Hence no formalism for describing and studying entropy is complete, if it 
cannot also find expression in a type III context. In these notes we will provide a formalism for describing (relative) entropy for the 
most general quantum systems. Our approach is to describe relative entropy in terms of modular dynamics, for which a common input 
stemming from the concept of Radon-Nikodym derivatives is crucial. As we shall show in section \ref{sec4}, the theory achieved dovetails 
well with existing concepts of relative entropy \cite{ara1, ara2}, and also allows for a ``density based'' description which faithfully 
mimics the classical formula. Then in section \ref{sec5} we use the theory thus developed as a guideline for introducing a concept of 
entropy for single states of type III von Neumann, before concluding by indicating the way forward. As shall be seen, the definition for 
entropy achieved in section \ref{sec5} harmonises perfectly with the above-mentioned Orlicz space formalism, and is a natural extension 
of the descriptions given in \cite{[2]}. We emphasize that the aforementioned extension demands a regularization procedure which has recently been shown to fit the operator algebraic approach to Quantum Field Theory very naturally, cf \cite{LM}.

\section{Boltzmann's H-functional and (classical) entropy}
Let $\Gamma$ be  a phase space associated with a system. We fix a reference measure $\lambda$ on $\Gamma$, usually it will be the Lebesgue measure. A function $f$ such that $f \in \{ g; g \in L^1(\Gamma, d\lambda), g\geq 0, \int_{\Gamma} g d\lambda = 1 \}\equiv \mathfrak{S}$ defines a probability measure $d\mu = f d \lambda$. In the Boltzmann theory, such an $f$ has the interpretation of velocity distribution function, cf \cite{thompson} see also \cite{ruelle} . On the other hand we note that $g \in \mathfrak{S}$ can be written as 
\begin{equation}
g = \frac{d\mu}{d\lambda}
\end{equation}
where $\frac{d\mu}{d\lambda}$ stands for the Radon-Nikodym derivative. Hence, the Boltzmann $H$-functional can be written as
\begin{equation}
H(g) \equiv \int g \log (g) d\lambda = \int \frac{d\mu}{d\lambda} \log\left(\frac{d\mu}{d\lambda}\right) d\lambda = \mu\left(\log\left(\frac{d\mu}{d\lambda}\right)\right),
\end{equation}
provided that the above integrals exist.
In \cite{[1]}, \cite{[2]}, \cite{[3]} we have argued that for states (probability measures) $\mu$ such that $\frac{d\mu}{d\lambda} \in L\log(L+1) \cap L^1$, the functional $H(\cdot)$ is well defined.
\begin{remark}
As the (classical) continuous entropy $S$ differs from the functional $H$ only by sign, the above means that the entropy
$S(\frac{d\mu}{d\lambda})$ is well defined if $\frac{d\mu}{d\lambda} \in L\log(L+1)\cap L^1$.
\end{remark}

Let $\mu$ and $\nu$ be probability measures over a set $X$, and assume that $\mu$ is absolutely continuous with respect to $\nu$. The relative entropy (also known as Kullback-Leibner divergence) is defined as
\begin{equation}
\label{1.3}
S(\mu|\nu) = \int_X \log \left(\frac{d\mu}{d\nu}\right)d \mu = \int_X \frac{d\mu}{d\nu} \log\left(\frac{d\mu}{d\nu}\right) d\nu \equiv \left\langle\log\frac{d\mu}{d\nu}\right\rangle_{\mu},
\end{equation}
provided that the integrals in the above formulas exist, where $\frac{d\mu}{d\nu}$ is the Radon-Nikodym derivative of $\mu$ with respect to $\nu$. Assume additionally that $\nu$ (so also $\mu$ is absolutely continuous with respect to the reference measure $\lambda$. Then
\begin{equation}
\frac{d\mu}{d\nu} = \frac{d\mu}{d\lambda} \cdot \frac{d\lambda}{d\nu},
\end{equation}
and under some additional assumptions one has the more familiar formula for the relative entropy
\begin{equation}
\label{1.5}
S(\mu|\nu) = \int_X p \log \frac{p}{q} d\lambda,
\end{equation}
where $p = \frac{d\mu}{d\lambda}$ and $q = \frac{d\nu}{d\lambda}$.

Intuitively, it is easily seen that for a discrete case, the entropy of a random variable $f$ on $X$ with a probability distribution $p(x)$ is related to how much $p(x)$ diverges from the uniform distribution on the support of $f$. In particular, putting $q=1$ in the formula \ref{1.5} one gets
\begin{equation}\label{1.6}
S(\mu|\tau) = H(p),
\end{equation} 
where the (non-normalized) functional $\tau$ is defined by the reference measure $\lambda$. As ``uniformity'' can be related to the ``most'' chaotic state (each microstate is equally probable), the basic property of statistical entropy expressing how far the given state is from the most chaotic, is recovered.

To clarify this point as well as to gain some intuition for a noncommutative generalization, we turn to the algebraic approach to the just defined concepts.
For a fixed measure space $(X, \Sigma, \lambda)$, let $L^{\infty}(X, \Sigma, \lambda) \equiv L^{\infty}$ denote the set of all $\lambda$-measurable, essentially bounded functions on $X$. The absolute continuity of $\mu$ with respect to $\lambda$ is equivalent to the condition that $\mu$ can be regarded as a normal functional on $L^{\infty}(X, \Sigma, \lambda)$, cf \cite{Ber} Theorem 1 , p. 167.
Since $L^{\infty}(X, \Sigma, \lambda)$ is the prototype of abelian von Neumann algebras, one can rewrite definitions as well as the basic properties of the above concepts in (abelian) von Neumann algebraic terms.

To this end, let $\vartheta_{\mu}(f) = \int_X f \cdot (\frac{d\mu}{d \lambda}) d\lambda(x)$ denote the functional over $L^{\infty}(X, \Sigma, \lambda)$, for the reference measure $\lambda$. In particular, the trace $\tau$ over $L^{\infty}(X, \Sigma, \lambda)$ is given by
\begin{equation}
\tau(f) = \int_X f d\lambda(x).
\end{equation}

It is worth pointing out that the existence of such a trace affords the possibility of dealing with uniform distribution (as was indicated above). In other words, such existence affords the possibility of discussing the relation between entropy and relative entropy!
Consequently, the entropy formula can be given as
\begin{equation}
S(\mu) \equiv S(\vartheta_{\mu}) = \tau\left(\left(\frac{D\vartheta_{\mu}}{D\tau}\right) \log \left(\frac{D\vartheta_{\mu}}{D\tau}\right)\right) \equiv
\int_X \left(\frac{D\vartheta_{\mu}}{D\tau}\right) \log \left(\frac{D\vartheta_{\mu}}{D\tau}\right) d\lambda(x) 
\end{equation}
$$ = \left\langle \log\left(\frac{D\vartheta_{\mu}}{D\tau}\right)\right\rangle_{\mu},$$
while the relative entropy formula reads
\begin{equation}\label{commcase}
S(\vartheta_{\mu}|\vartheta_\nu) = \left\langle \log\left(\frac{D\vartheta_{\mu}}{D\vartheta_\nu}\right)\right\rangle_{\mu},
\end{equation}
where $\frac{D\vartheta_{\mu}}{D\vartheta_\nu}$ stands for the Radon-Nikodym derivative of functional $\vartheta_{\mu}$
with respect to the functional $\vartheta_\nu$, see the next section.

\begin{remark} \textit{Classical equilibrium thermodynamics.}
\label{classical thermodynamics}
To get some better intuition, let us consider specific case, when the velocity distribution function $\frac{d\mu}{d\lambda}$ is given by Maxwell-Boltzmann distribution
\begin{equation}
\frac{d\mu}{d\lambda} = Z e^{- \beta H} = e^{log Z - \beta H} \equiv e^K,
\end{equation}
where $Z$ is the normalization constant, $\beta >0$ (usually interpreted as ``the inverse temperature''), and $H$ is the Hamiltonian of the system under consideration.

For such cases, the above formulas for entropies read
\begin{equation}
\label{1.11}
S(\frac{d\mu}{d\lambda}) \equiv S(\frac{D\upsilon_{\mu}}{D\tau}) = \langle \log(e^K)\rangle_{\mu} = \langle K \rangle_{\mu},
\end{equation}
while for the relative entropy of $\frac{d\mu}{d\lambda} = Z_1 e^{- \beta_1H_1} \equiv e^{K_1}$, $\frac{d\nu}{d\lambda} = Z_2 e^{- \beta_2H_2} \equiv e^{K_2}$, one has
\begin{equation}
S\left(\frac{d\mu}{d\lambda}|\frac{d\nu}{d\lambda}\right) = \langle \log \frac{e^{K_1}}{e^{K_2}}\rangle_{\mu} = \langle K_1\rangle_{\mu} - \langle K_2\rangle_{\mu}.
\end{equation}
It is \textit{important to note that (\ref{1.11}) is in perfect agreement with the second law of thermodynamics}; see section 32 in \cite{chin}.
The above formulas can be rewritten as
\begin{equation}
S\left(\frac{d\mu}{d\lambda}\right) = -i\frac{d}{dt}\langle e^{itK}\rangle_{\mu}|_{t=0}
\end{equation}
and
\begin{equation}
S\left(\frac{d\mu}{d\lambda} | \frac{d\nu}{d\lambda}\right) = -i \frac{d}{dt}\langle e^{itK_1} e^{-itK_2}\rangle_{\mu}|_{t=0}
\end{equation}

As it will be seen in the next Sections, the above recipe can easily be generalized and quantized.
\end{remark}
To clarify the significance of derivatives and to proceed with our exposition we need some preliminaries, which for the reader's convenience will be given in a separate section.


 
\section{Algebraic preliminaries}
As the concepts of entropy and relative entropy involve Radon-Nikodym derivatives, for the reader convenience, we here provide the relevant material on noncommutative Radon-Nikodym and cocycle derivatives, thus making
our exposition self-contained. The theory of such cocycles goes back to \cite{Con1}, \cite{Con2}, \cite{ConTak}.
In particular, Connes proved, see \cite{Con1}
\begin{theorem}
\label{connes}
Let $\qM$ be a von Neumann algebra and $\phi$, $\psi$ faithful semifinite normal weights on $\qM$. Then there exists a $\sigma$-strongly continuous one parameter family $\{u_t\}$ of unitaries in $\qM$ with the following properties:
\begin{itemize} 
\item $u_{t+t^{\prime}} = u_t \sigma_t^{\phi}(u_{t^{\prime}}), \mbox{ for all } t, t^{\prime} \in \Rn,$
\item
\begin{equation}
\label{2.1a}
\sigma^{\psi}_t (x) = u_t \sigma_t^{\phi}(x)u_t^*, \mbox{ for all } x \in \qM, t \in \Rn,
\end{equation}
\item a unitary $u\in \qM$ satisfies $\psi(x)=\phi(uau^*)$ for all $x\in \qM$, if and only if $u_t=u^*\sigma^\phi_t(u)$ for all $t\in \mathbb{R}$,
\end{itemize}
where $\sigma_t^{\varphi}$ ($\sigma_t^{\psi}$) stands for the modular evolution determined by $\varphi$ ($\psi$ respectively).
\end{theorem}
\begin{definition}
The family of unitaries described by the above theorem is called the cocycle derivative of $\varphi$ with respect to
$\psi$ and is denoted by
\begin{equation}
(D\varphi : D\psi)_t = u_t.
\end{equation}
\end{definition}
To understand fully the next remark we need, cf \cite{BR}, Theorem 5.3.10.
\begin{theorem} (\textit{Takesaki})
\label{takesaki}
Let $\qM$ be a von Neumann algebra, and $\omega$ a normal state on $\qM$. The following are equivalent:
\begin{enumerate}
\item $\omega$ is a faithful as a state on $\pi_{\omega}(\qM)$, i.e. there exists a projector $E \in \qM \cap\qM^{\prime}$ such that $\omega(\jed - E) = 0$ and $\omega|_{\qM E}$ is a faithful state.
\item There exists a $\sigma$-weakly continuous one-parameter group $\sigma$ of ${}^*$-automorphisms of $\qM$ such that $\omega$ is $\sigma$-KMS state.
\end{enumerate}
Moreover, $\sigma|{\qM E}$ is the modular group of $\qM E$ associated with $\omega$.
\end{theorem}
This theorem legitimises the application of KMS theory to our approach to quantum entropy. In other words, our scheme is related to \textit{quantum equilibrium thermodynamics}. Now we are in position to present
\begin{remark}
\label{2.4a}
We note that $u_0 = \jed$ (see the proof of Theorem 3.3, Chapter  VIII in \cite{Tak2}). Further, let us take (formally) a derivation of (\ref{2.1a}) at $t=0$. Then, denoting the infinitesimal generator of $\sigma_t^{\psi}$ ($\sigma^{\varphi}_t$) by $L^{\psi}$ ($L^{\varphi}$ respectively) one gets
\begin{equation}
\label{derivation}
L^{\psi}(x) = \left.\frac{du_t}{dt}\right\bracevert_{t=0} x + L^{\varphi} (x) + x \left(\left.\frac{du_t}{dt}\right\bracevert_{t=0}\right)^*,
\end{equation} 
or equivalently
\begin{equation}
L^{\psi}(x) - L^{\varphi} (x) = \left.\frac{du_t}{dt}\right\bracevert_{t=0} x + x \left(\left.\frac{du_t}{dt}\right\bracevert_{t=0}\right)^*.
\end{equation}
Theorem \ref{takesaki} implies that the modular evolution for a fixed faithful normal state $\varphi$ on $\qM$, can be interpreted as Hamilton type dynamics for the equilibrium (KMS) state on $\qM$.
This means that the derivative of $u_t$ at $t=0$ determines the difference of two ``equilibrium'' type generators $L^{\psi}$ and $L^{\varphi}$.
The important point to note here, is the fact that in general, $L^{\psi}$ and $L^{\varphi}$ are unbounded derivations. Thus, the equality (\ref{derivation}) is not well defined for each $x$. This clearly indicates that derivatives of $u_t$ should be studied carefully, and this will be done in the ensuing sections.

To say more, let $\psi$ be a perturbed $\varphi$-state, so $\psi \equiv \varphi^P$; for all details see section 5.4.1 in \cite{BR}. In particular, for $P\in \qM$ there exists an explicit form of $u_t$. Furthermore, it is easy to note that $L^{\varphi^P}x - L^{\varphi}x = i[P,x]$, which is well defined. Consequently, comparing two states which differ by finite a energy perturbation, does not lead to any problem.

Finally, we note that KMS states can be characterized by passivity, see \cite{PW} and/or section 5.4.4 in \cite{BR}.
We remind that among other things passivity ensures compatibility with the second law of thermodynamics. Therefore, our scheme based on Tomita-Takesaki theory, seems to be a natural quantization  
of the classical case presented in Remark \ref{classical thermodynamics}.
\end{remark}

The Radon-Nikodym theorem used in the previous section has generalizations to general von Neumann algebras. The first generalization,
for traces,  is extracted from Pedersen's book \cite{Peder}, see Theorem 5.3.11 and remarks in 5.3.12.

\begin{theorem}
\label{2.1}
Let $\tau$ be a normal semifinite trace over $\mathfrak{M}$. For each normal semifinite weight $\psi$ on $\mathfrak{M}$ 
which is absolutely continuous with respect to $\tau$ in the sense that for any $a \in \mathfrak{M}$, the fact that $\tau(a^*a) = 0$ implies $\psi(a^*a)=0$, there exists a unique positive operator $h$ on $\cH_{\tau}$ ($\cH_{\tau}$ is GNS space for $(\mathfrak{M}, \tau))$ such that
\begin{equation}
\psi(x) = \tau(h x).
\end{equation}
\end{theorem}

For a general von Neumann algebra $\mathfrak{M}$ and two normal faithful semifinite weights such that one dominates the other one has (see Theorem VIII.3.17 in \cite{Tak2})

\begin{theorem}
\label{2.2}
For a pair $\vartheta$, $\psi$ of faithful semifinite normal weights on $\qM$, the following conditions are equivalent:
\begin{enumerate}
\item There exists $M>0$ such that
\begin{equation}
\vartheta(x) \leq M \psi(x), \quad x \in \qM_+,
\end{equation}
\item The cocycle derivative $({D\psi}\colon{D\vartheta})_t \equiv u_t$ can be extended to an $\qM$-valued $\sigma$-weakly continuous bounded function on the horizontal strip $\overline{D}_{\frac{1}{2}} = \{ z \in \Cn; -\frac{1}{2}\leq Im(z) \leq 0 \}$, which is holomorphic in the interior of the strip.
\end{enumerate}
If these conditions hold, then
\begin{equation}
\vartheta (x) = \psi(u^*_{-\frac{i}{2}} x u_{-\frac{i}{2}}), \quad x\in \{ \sum_i^n y^*_ix_i; \quad x_i,y_j \in n_{\psi}\},
\end{equation}
where $n_{\psi} = \{ x\in \qM; \psi(x^*x) < \infty \}$.
\end{theorem}
\begin{remark}
We emphasize that a domination of one weight by another is a stronger property than ``absolute continuity'' described in Theorem \ref{2.1}, but the domination condition is in the same vein as the condition of absolute continuity. Also notice from part (2) of the above theorem, that $|u_{-i/2}^*|^2$ in a very real sense fulfills the role of the ``density'' of $\psi$ with respect to $\vartheta$. 
\end{remark}
One may ask whether there is a relation, based on the Connes characterization of unitary Radon-Nikodym cocycles, between cocycle derivatives and the relative modular operator. More precisely, see \cite{ara1}, \cite{Araki1}, let $\phi$, $\vartheta$ be normal semifinite weights on $\qM$, and $\phi$ be faithful. Then
\begin{equation}
\label{2.4}
u_t \equiv \left({D \vartheta}\colon {D\phi}\right)_t = \Delta^{it}_{\vartheta, \phi} \Delta_{\phi}^{-it}.
\end{equation}

In particular, if $\qM$ is semifinite von Neumann algebra, $\psi$ and $\vartheta$ faithful semifinite normal weights, $\tau$ a faithful, normal semifinite trace on $\qM$, then one has (see \cite{connes}, p.470) that 
there exist positive operators affiliated with $\qM$
such that $\psi(x) = \tau(\varrho_{\psi} x)$, $\vartheta(x) = \tau(\varrho_{\vartheta}x)$ for each $x \in \qM$, and
\begin{equation}
\left({D\vartheta}\colon {D\psi}\right)_t = \varrho_{\vartheta}^{it}\varrho_{\psi}^{-it}.
\end{equation}
Hence, on applying this equality to the abelian von Neuma algebra $L^{\infty}$ (cf. the discussion at the end of the previous section), one has 
\begin{equation}
\frac{d}{dt}\left(D\varphi_{\mu} \colon D\varphi_{\nu}\right)|_{t = 0} = i \log f_{\mu} - i\log f_{\nu},
\end{equation}
where $\mu = f_{\mu}d\lambda$, $\nu = f_{\nu} d \lambda$, and $f_{\mu}>0$, $f_{\nu}>0$.  Thus
\begin{equation}
-i \left\langle \frac{d}{dt}\left.\left(D\varphi_{\mu} \colon D\varphi_{\nu}\right)\right\bracevert_{t=0}\right\rangle_{\mu} = \int f_{\mu} \log\frac{f_{\mu}}{f_{\nu}} d\lambda = \int\left(f_{\mu}\log f_{\mu} - f_{\mu}\log f_{\nu}\right)d \lambda,
\end{equation}
which is in perfect agreement with the definition of the relative entropy, cf. formula \ref{1.3}.

\vskip 1cm

We remind the reader that the proper basic structure for a description of large quantum systems, is a von Neumann algebra of type III. In other words, one is forced to deal with a von Neumann algebra which is not equipped with a nontrivial trace. \textit{Consequently, to be able to study entropy, access to the type of functional calculus required for an effective description of uniform distribution would be a powerful tool, which can be accessed by passing to a larger super-algebra, i.e. to the crossed-product $\cM$.} It is in this larger super-algebra that we have access to the functional calculus for $\tau$-measurable operators. If $\qM$ together with a canonical faithful normal semifinite weight $\omega$ is given on a Hilbert space $\cH$, then $\cM$ is the von Neumann algebra on the Hilbert space $L^2(\bR, \cH)$ generated by the following operators:
\begin{equation}
(\pi(x)\xi)(t) = \sigma_{-t}(x) \xi(t), \\ \xi \in L^2(\bR, \cH), t \in \bR, x \in \qM,
\end{equation}
\begin{equation}
(\lambda(s)\xi)(t) = x(t - s), \\ \xi \in L^2(\bR, \cH), t \in \bR, x \in \qM,
\end{equation}
where $\sigma_t=\sigma^\omega_t$ stands for the modular automorphism.
\begin{remark}
\begin{enumerate}
\item $\qM$ can be identified with its image $\pi(\qM)$ in $\cM$.
\item If $\qM$ is type III then $\cM$ is a semifinite. Thus, on $\cM$ there is a semifinite normal faithful trace!
\end{enumerate}
\end{remark}
We wish to close this section with a deep result of Haagerup, see \cite{haare} Theorem 4.7 or/and \cite{terp} pp. 26-27.
Let $\psi$, $\vartheta$ be normal, faithful semifinite weights on $\qM$. $\tilde{\psi}$ and $ \tilde{\vartheta}$ stand for the corresponding dual weights on $\cM$. Then, for any $t \in \bR$
\begin{equation}
\label{2.10}
\left({D\tilde{\psi}}\colon {D\tilde{\vartheta}}\right)_t = \left({D\psi}\colon {D\vartheta}\right)_t.
\end{equation}

\section{The von Neumann entropy and Dirac's formalism.}
In Dirac's formalism, a (small) quantum system is described by an infinite dimensional Hilbert space $\cH$ and the von Neumann algebra $B(\cH)$. A normal state $\psi$ on $B(\cH)$ has the form $\psi(a) = \Tr \varrho_{\psi} a$ where $\varrho_{\psi}$ is a positive trace class operator, with trace equal to $1$, i.e. $Tr \varrho_{\psi} = 1$. Here the set of states $\mathfrak{S}$ is given by $\mathfrak{S} = \{ \varrho \in B(\cH); \varrho^* = \varrho, \varrho \geq 0, \Tr \varrho = 1\}$.  Applying the non-commutative Radon-Nikodym theorem, see Theorem 1, pp. 469-470 in \cite{connes}, one has
\begin{equation}
\varrho_{\psi}^{it} = (D\psi : D\Tr)_t
\end{equation}
We remind that von Neumann entropy $S(\varrho_{\psi})$ was defined as
\begin{equation}
S(\varrho_{\psi}) = \Tr (\varrho_{\psi} \log \varrho_{\psi}).
\end{equation} 

This definition can be rewritten in Radon-Nikodym terms in the following way
\begin{equation}
S(\varrho_{\psi}) = -i \Tr \left(\varrho_{\psi}\left.\frac{d}{dt} (D\psi : D\Tr )\right\bracevert_{t= 0}\right)  \equiv -i \psi \left(\left.\frac{d}{dt}(D\psi : D\Tr)\right\bracevert_{t=0} \right),
\end{equation}
and, for $\psi(\cdot) = \Tr (\varrho_{\psi} \cdot)$, $\varphi(\cdot) = \Tr(\varrho_{\varphi} \cdot)$
\begin{equation}
\label{3.4}
S(\psi|\varphi) = \Tr\left(\varrho_{\psi} \log \varrho_{\psi} - \varrho_{\psi} \log \varrho_{\varphi}\right) = -i \psi\left(\left.\frac{d}{dt}(D\psi \colon D\varphi)\right\bracevert_{t=0} \right),
\end{equation}
where we assumed that the states are faithful.
\begin{remark}
As was pointed out at the end of \cite[Section 6]{[2]}, within Dirac's formalism the Orlicz space scheme for selecting ``good'' states with well defined entropy, agrees with the standard approach to elementary quantum mechanics. Specifically in this setting the space $L^1\cap L(\log(L+1))(B(\cH))$ is precisely the trace class operators $L^1(B(\cH))$. In fact for $B(\cH)$, all noncommutative measurable operators are already bounded; for details see cf \cite{[2]}, and \cite{Maj}. This behaviour is not unexpected as on the one hand Dirac's formalism is designed for small systems, and on the other hand, restricting to $B(\cH)$, noncommutative integration theory is oversimplified. The entropy for large systems will be examined in the next section.
\end{remark}

\section{General quantum case}\label{sec4}

Let us consider a  general quantum system and let $\mathfrak{M}$ be a von Neumann algebra associated with the system. In general, for large systems, $\mathfrak{M}$ is a type III von Neumann algebra. Let $\omega$ be a normal semifinite faithful weight on $\mathfrak{M}$.
The weight $\omega$ will play the role of a noncommutative probability reference measure. We denote by $\cM$ the cross product of $\mathfrak{M}$ associated with the modular morphism $\sigma_{\omega}$ produced by $\omega$, cf. Section 2. By $\tilde{\omega}$ we denote the dual (and hence normal semifinite faithful) weight on $\cM$, and $\tau$ stands for the canonical trace on $\cM$.
We remind that the modular automorphism $\tilde{\sigma}$ produced by the dual weight $\tilde{\omega}$ has the form $\tilde{\sigma}_t(\cdot) = \lambda(t) \cdot \lambda(t)^*$ - for details see \cite{Tak}, \cite{terp} and \cite{LM}. By Stone's theorem one has $\lambda(t) = h^{it}$. We note that $\log h$ can be identified with $-i\frac{d}{dt}(D\tilde{\omega}:{D\tau})|_{t=0}$ where $\tau$ is the canonical trace on $\cM$. A definition of relative entropy based on this derivative is therefore a compelling proposition. We pause to review Araki's definition of relative entropy before pursuing this matter further.

Let $\qM$ be a $\sigma$-finite von Neumann algebra in standard form described above, and let $\psi$ and $\phi$ be two faithful normal states with unit vector representatives $\Psi, \Phi \in \mathcal{H}$. The basic theory of Tomita-Takesaki theory easily extends to show that the densely defined anti-linear operator $S_{\phi,\psi}(a\Psi)=a^*\Phi$ is in fact closable. The operator $\Delta_{\phi,\psi}$ is then defined to be the modulus of the closure of $S_{\phi,\psi}$. In the same way that the ``standard'' modular operator may be used to generate the modular automorphism group of a given state, this operator in a very real sense encodes the manner in which the dynamics determined by the modular automorphism group of one state, differs from other. Using this fact, Araki then defined the relative entropy of $\psi$ and $\phi$ to be $-\langle \Psi,\log(\Delta_{\phi,\psi}) \Psi\rangle$. We refer the interested reader to \cite{ara1, ara2} and the references therein for a survey of the basic properties of this entropy. Despite the success of Araki's approach, we prefer a definition based on the derivative of the Connes cocycle, since on the one hand it is more overtly based on modular dynamics, and on the other it more easily allows for the incorporation of crossed product techniques in the study of this entropy -- as we shall subsequently see. 

A version of such a theorem appears in the book of Ohya and Petz \cite[Theorem 5.7]{OP}, where the authors showed that we infact have that  $S(\phi|\psi)=\lim_{t\to 0}\frac{-i}{t}\vartheta[(D{\psi} \colon D{\phi})_t-\I]$ but under the assumption that $-\langle \Psi,\log(\Delta_{\phi,\psi}) \Psi\rangle$ is finite. Some of the crucial facts which help to establish this link, is the fact that the Connes cocycle derivative $(D{\psi} \colon D{\phi})_t$ may be described in terms of $\Delta_{\phi,\psi}$ and $\Delta_\phi$ (see Appendix B of \cite{Araki1} for details). Another is that any normal state $\vartheta$ on a $\sigma$-finite von Neumann algebra $\qM$ in standard form, must have a vector representative \cite[Theorem 2.5.31]{BR}. We will show that if the above limit is modified in a manner which modulates the spectral behaviour of $\Delta_{\psi,\vartheta}$ near 0, one gets general equality of the two approaches. 

In anticipation of Theorem \ref{ar-ent}, we therefore propose the following definition:

\begin{definition}\label{gqdef}
Let $\psi$, $\vartheta$ be faithful normal states on $\qM$. We define the relative entropy $S(\vartheta|\psi)$ to be $S(\vartheta|\psi)=\lim_{\epsilon\searrow 0}\lim_{t\to 0}\frac{-i}{t}\psi[((D{\psi} \colon D{\vartheta})_t-\I)\chi_{[\epsilon, \infty)}(\Delta_{\vartheta,\psi})]$.
\end{definition}

We proceed with proving the promised theorem. Our proof is quite different from the one used by Ohya and Petz. In particular we are able to apply the Dominated Convergence theorem directly to the functions $|\frac{1}{t}(\lambda^{it}-1)|$, rather than to $|\log\lambda - \frac{1}{it}(\lambda^{it}-1)|$.  

\begin{theorem}\label{ar-ent} 
Let $\qM$ be a $\sigma$-finite von Neumann algebra in the standard form described above, and let $\psi$ and $\phi$ be two faithful normal states with unit vector representatives $\Psi, \Phi \in \mathcal{H}$. Then, allowing for $\infty$ as a possible value, the limit $\lim_{\epsilon\searrow 0}\lim_{t\to 0}\frac{-i}{t}\psi[((D{\psi} \colon D{\phi})_t-\I)\chi_{[\epsilon, \infty)}(\Delta_{\phi,\psi})]$ always exists and equals $-\langle\log(\Delta_{\phi,\psi})\Psi,\Psi\rangle$. In the case where this quantity is finite, we even have that 
$\lim_{t\to 0}\frac{-i}{t}\psi[(D{\psi} \colon D{\phi})_t-\I] = -\langle\log(\Delta_{\phi,\psi})\Psi,\Psi\rangle$. 
\end{theorem}

\begin{proof}
Recall that Araki's definition of entropy is $$S(\phi|\psi) = -\langle\log(\Delta_{\phi,\psi}) \Psi, \Psi\rangle$$where the latter term is understood to be $$-\int_0^\infty\log(\lambda)\,d\langle e_\lambda \Psi, \Psi\rangle$$(here $\lambda\to e_\lambda$ is the spectral resolution of $\Delta_{\phi,\psi}$). As Araki points out, the value of this integral is either real (in the case where $\log$ is integrable), or $\infty$ otherwise. To see this, notice that if $\log$ is not integrable, the fact that $0\leq \log(\lambda)\leq \lambda$ on $[1,\infty)$ ensures that $$0\leq \int_1^\infty\log(\lambda)\,d\langle e_\lambda \Psi, \Psi\rangle \leq \int_0^\infty\lambda\,d\langle e_\lambda \Psi, \Psi\rangle=
\langle \Delta_{\phi,\psi} \Psi, \Psi\rangle<\infty,$$and hence that we nevertheless still have that $\log$ is integrable on $[1,\infty)$. Therefore, any nonintegrability of $\log$ must be derived from the fact that $-\int_0^1\log(\lambda)\,d\langle e_\lambda \Psi, \Psi\rangle=\infty$.

First, suppose that $\log$ is integrable. For any $t$ and any $\lambda>0$, we have that $$\left|\frac{1}{t}(\lambda^{-it}-1)\right| \leq \left|\frac{1}{t}(\lambda^{-it/2}-1)(\lambda^{-it/2}+1)\right| \leq \left|\frac{2}{t}(\lambda^{-it/2}-1)\right|.$$ Carrying on inductively leads to the
conclusion that $|\frac{1}{t}(\lambda^{-it}-1)| \leq |\frac{2^k}{t}(\lambda^{-it/2^k}-1)|$ for any $k\in \mathbb{N}$. If now we let
$k\to \infty$, we obtain the inequality $|\frac{1}{t}(\lambda^{it}-1)| \leq |\log(\lambda)|$, which holds for any $t>0$ and any
$\lambda>0$. Hence we may apply the dominated convergence theorem to see that for any sequence $(t_n)$ converging to 0, we have that
\begin{eqnarray*}
\lim_{n\to \infty}\frac{-i}{t_n}\langle(\Delta_{\phi,\psi}^{-it_n}-\I)\Psi, \Psi\rangle &=& \lim_{n\to \infty}\frac{-i}{t_n}\int_0^\infty(\lambda^{-it_n}-1)\,d\langle e_\lambda \Psi, \Psi\rangle\\
&=& -\int_0^\infty\log(\lambda)\,d\langle e_\lambda \Psi, \Psi\rangle.
\end{eqnarray*}
This fact is enough to enable us to conclude that $$\lim_{t\to 0}\frac{-i}{t}\langle (\Delta_{\phi,\psi}^{-it}-\I)\Psi, \Psi\rangle = -\int_0^\infty\log(\lambda)\,d\langle e_\lambda \Psi, \Psi\rangle.$$

We now consider the general case and leave it as an exercise to show that we will for any $\epsilon>0$ have that 
$$\frac{-i}{t}\psi[((D\psi:D\phi)_t-{\mathds{1}})\chi_{[\epsilon,\infty)}(\Delta_{\phi,\psi})] = \frac{-i}{t}\int_\epsilon^\infty(\lambda^{-it}-1)\,d\langle e_\lambda\Psi,\Psi\rangle.$$Since 
 $\lambda\mapsto\log(\lambda)$ is bounded and continuous on $[\epsilon,1]$, an entirely analogous argument shows that we will for any sequence $(t_n)$ converging to 0 have that 
$$\lim_{n\to\infty}\frac{-i}{t_n}\psi[((D\psi:D\phi)_{t_n}-{\mathds{1}})\chi_{[\epsilon,\infty)}(\Delta_{\phi,\psi})] = -\int_\epsilon^\infty\log(\lambda)\,d\langle e_\lambda\Psi,\Psi\rangle.$$Since 
$\int_1^\infty\log(\lambda)\,d\langle e_\lambda\Psi,\Psi\rangle$ is finite and $\log$ decreasing without bound on $[0,1]$, it is palpably clear that even for the case where the the full integral 
diverges, we have that $\lim_{\epsilon\searrow}\int_\epsilon^\infty\log(\lambda)\,d\langle e_\lambda\Psi,\Psi\rangle = -\int_0^\infty\log(\lambda)\,d\langle e_\lambda\Psi,\Psi\rangle$ which then proves 
the theorem.
\end{proof}

\begin{remark}\label{rem-4.3}
If $\qM$ is commutative, then $\qM = L^{\infty}(X, \mu)$, in which case $\psi$ and $\phi$ correspond to positive measures on $X$. In particular (cf Theorem \ref{connes}) there exists a Radon-Nikodym derivative $h=\frac{d\psi}{d\phi}$, and $h^{it} = (D\psi : D\phi)_t$. Therefore, the definition of classical relative entropy is also stemming from Definition \ref{gqdef}. Finally, the definition of relative entropy for Dirac's formalism also follows from Definition \ref{gqdef} (cf formula (\ref{3.4})).
\end{remark}

To say more, we are going to invoke some results from the theory of $L^p$-spaces associated with von Neumann algebras. We note, cf. \cite{terp}, Theorem 36, that
\begin{equation}
\left( \lambda(\qM), L^2(\qM), J, L^2(\qM)_+\right)
\end{equation}
is the standard form of $\qM$, where the right action $\lambda(\cdot)$ is defined as $\lambda(a)a = a, a \in L^2(\qM)$ and $J$ denotes the conjugate isometric involution $a \mapsto a^*$ of $L^2(\qM)$

Next, we note, cf. \cite{terp} Proposition 4,  that there is a bijection $\phi \mapsto h_{\phi}$ of the set of all normal semifinite weights on $\qM$ onto the set of all positive selfadjoint operators $h$ affiliated with $\cM$, and satisfying $\theta_s(h) = e^{-s}h$ for any $s \in \Rn$. As we wish to deal with $\tau$-measurable operators we must restrict ourselves to normal functionals on $\qM$. Then, the mapping $\phi \mapsto h_{\phi}$
is an isometry of $\qM_*$ onto $L^1(\qM)$. Consequently, fixing $\phi \in \qM_{*,+}$, one gets $h_{\phi} \in L^1(\qM)_+$. In particular, $h^{\frac{1}{2}}_{\phi} \in L^2(\qM)_+$, and
\begin{equation}
\phi(x) = tr(h^{\frac{1}{2}}_{\phi} x h^{\frac{1}{2}}_{\phi}) \equiv \langle h^{\frac{1}{2}}_{\phi},x h^{\frac{1}{2}}_{\phi}\rangle_{L^2(\qM)},
\end{equation}
where $tr$ stands for a linear functional (having the trace property) on $L^1(\qM)$, see Definition II.13 and Proposition II.21 in \cite{terp}.
In other words, $h^{\frac{1}{2}}_{\phi}$ is a vector in the natural cone $L^2(\qM)_+$, and this vector represents the state $\phi$.

Using the above framework, the proposed definition of entropy may be written as the claim that $S(\psi|\phi)=  \lim_{t\to 0} \frac{-i}{t}tr(h_\psi^{1/2}[(D{\psi} \colon D{\phi})_t-\I]h_\psi^{1/2})$ whenever the limits exists, with  $S(\psi|\phi)= \infty$ otherwise. The next objective in this section, is to show that this definition can very concretely be reformulated in a manner which is a faithful noncommutative analogue of the classical formula presented in Equation \ref{1.5}. However for this we will need some additional technology which we now review.

The first factor that suggests that such a formula may well be within reach is the fact that in the above framework we have that
$$(D{\vartheta} \colon D{\psi})_t = h_\vartheta^{it}h_\psi^{-it}$$where $h_\vartheta=\frac{D\widetilde{\vartheta}}{D\tau}$ and 
$h_\psi=\frac{D\widetilde{\psi}}{D\tau}$. To see that the above claim is true, we may use Haagerup's result and the cocycle chain rule to see that 
$$\I = (D{\widetilde{\psi}} \colon D\widetilde{\psi})_t = (D{\widetilde{\psi}} \colon D\tau)_t (D\tau \colon D\widetilde{\psi})_t.$$Equivalently
$$(D\tau \colon D\widetilde{\psi})_t=(D{\widetilde{\psi}} \colon D\tau)_t^{-1}.$$But from section 2 we know that $(D{\widetilde{\psi}} \colon D\tau)_t=h_\psi^{it}$. Hence $(D\tau \colon D\widetilde{\psi})_t = h_\psi^{-it}$. Since also $(D{\widetilde{\vartheta}} \colon D\tau)_t=h_\vartheta^{it}$, we may once again use Haagerup's result and the chain rule to see that
\begin{eqnarray}\label{alt-cocycle}
(D{\vartheta} \colon D{\psi})_t &=& (D{\widetilde{\vartheta}} \colon D\widetilde{\psi})_t\\
&=& (D{\widetilde{\vartheta}} \colon D\tau)_t (D\tau \colon D\widetilde{\psi})_t\nonumber\\
&=& h_\vartheta^{it}h_\psi^{-it}\nonumber
\end{eqnarray}

Another major factor to take into account is that $tr$ is only defined on $L^1(\qM)$. Thus, to proceed with our objective of developing a noncommutative version of formula (\ref{1.5}) we must show that in some sense $h_{\vartheta} \log h_{\vartheta} - h_{\vartheta} \log h_{\phi}$ is in $L^1$. As we shall see below, this can indeed be achieved in a limiting sense. Following Terp's arguments, see \cite{terp} Lemma II.19, we consider the function
\begin{equation} 
\label{terp'sfunction}
S^0 \ni \alpha \mapsto h_{\vartheta}^{\alpha} h_{\phi}^{1-\alpha} \in L^1(\qM),
\end{equation}
where obviously $h_{\vartheta}, h_{\phi} \in L^1(\qM)$, and $S$ is the closed complex strip 
$\{ \alpha \in \Cn; 0\leq \Re(\alpha) \leq 1 \}$ and $S^0$ stands for the corresponding open strip. Terp's Lemma II.19 easily adapts to show that the function (\ref{terp'sfunction}) is analytic on $S^0$. Taking the derivative, in the Banach space language, one gets that
\begin{equation}
\label{4.6}
\alpha \mapsto  h_{\vartheta}^{\alpha} \cdot \log h_{\vartheta} \cdot h_{\phi}^{1 - \alpha} - h_{\vartheta}^{\alpha} \cdot \log h_{\phi} \cdot h_{\phi}^{1 - \alpha} \in L^1(\qM)
\end{equation}
inside $S^0$. More importantly, the analyticity ensures that this derivative varies continuously with respect to $\alpha$ in 
$L^1$-norm. A fact which underlies the above very regular behaviour on this strip, is that for any $0<s$, $x^s\log(x)$ is very well behaved continuous function which is 0 at 0, and for which $x^s\leq x^s\log(x) \leq x^{s+1}$ whenever $x\geq e$. This fact can be used to show that for any positive $\tau$-measurable operator $g$, $g^s\log(g)$ will again be $\tau$-measurable. 

Using the above formula and letting $\alpha\to 1$, leads us to the promised noncommutative analogue of formula (\ref{1.5}). To understand how this is achieved, assume that $\alpha=s+it$ where $0< s < 1$. Then the fact that $$h_{\vartheta}^{\alpha} h_{\phi}^{1-\alpha}=h_{\vartheta}^{s}[h_{\vartheta}^{it} h_{\phi}^{-it}] h_{\phi}^{1-s}$$leads to the conclusion that 

\begin{eqnarray*}
-i\frac{d}{dt}(h_{\vartheta}^{s}(D{\vartheta}:D{\phi})_t h_{\phi}^{1-s})&=&-i\frac{d}{dt}(h_{\vartheta}^{s}[h_{\vartheta}^{it} h_{\phi}^{-it}] h_{\phi}^{1-s})\\
&=&\frac{d}{d\alpha} h_{\vartheta}^{\alpha} h_{\phi}^{1-\alpha}\\
&=& h_{\vartheta}^{\alpha} \cdot \log h_{\vartheta} \cdot h_{\phi}^{1 - \alpha} - h_{\vartheta}^{\alpha} \cdot \log h_{\phi} \cdot h_{\phi}^{1 - \alpha}.
\end{eqnarray*}
To see this, notice that in computing a limit of the form $\lim_{\Delta\alpha\to 0} \frac{f(\alpha+\Delta\alpha)-f(\alpha)}{\alpha}$, we may as well assume that $\Delta\alpha=i\Delta t$. In particular when computing the derivative anywhere along the line segment $0< s < 1$, $t=0$, we always have that
$$-i\frac{d}{dt}\left.(h_{\vartheta}^{s}(D{\vartheta}:D{\phi})_t h_{\phi}^{1-s})\right\bracevert_{t=0} = h_{\vartheta}^{s} \cdot \log h_{\vartheta} \cdot  h_{\phi}^{1 - s} - h_{\vartheta}^{s} \cdot \log h_{\phi} \cdot h_{\phi}^{1 - s}\in L^1(\qM).$$
With the groundwork having been done, we are now ready to present the promised result. 

The concept that is crucial in guaranteeing the validity of the theorem, is the following ordering defined by Takesaki and Connes. (See \cite[Definition 4.1]{ConTak}.)

\begin{definition}
For two normal weights $\vartheta$ and $\phi$ on a von Neumann algebra $\qM$, and some positive $\delta$, we say that $\vartheta\leq \phi(\delta)$ if the function $t\to (D\vartheta:D\phi)_t=u_t$ extends to an $\qM$-valued map $z\to u_z$ which is point to weak*-continuous and bounded
on the closed strip $\{z\in\mathbb{C}: -\delta\leq\Im(z)\leq 0\}$, and analytic on the open strip 
$\{z\in\mathbb{C}: -\delta<\Im(z)< 0\}$. 
\end{definition}

In the above ordering the case $\delta=\frac{1}{2}$ corresponds exactly to Theorem \ref{2.2}.

\begin{theorem}\label{thm4.3} 
Let $\qM$ be a $\sigma$-finite von Neumann algebra in the standard form and let $\vartheta$ and $\phi$ be two faithful normal states with unit vector representatives $h_\vartheta^{1/2}, h_\phi^{1/2} \in L^2(\qM)$. If $\phi\leq \vartheta(\delta)$, and $S(\phi|\vartheta)$ is finite, then limit
$$\lim_{s\nearrow 1} tr(h_{\vartheta}^{s}\cdot  \log h_{\vartheta} \cdot h_{\phi}^{1 - s} - h_{\vartheta}^{s} \cdot  \log h_{\phi} \cdot  h_{\phi}^{1 - s})$$exists and equals $S(\phi|\vartheta)$.
\end{theorem}

\begin{proof} First suppose that $S(\vartheta|\phi)$ is finite and let $\epsilon>0$ be given. This means that $\lim_{t\to 0}\frac{-i}{t}\vartheta[(D{\vartheta} \colon D{\psi})_t-\I]= \lim_{t\to 0}\frac{-i}{t}tr(h_{\vartheta}^{1/2} [h_{\vartheta}^{it} h_{\phi}^{-it}-\I]h_{\vartheta}^{1/2})$ exists. Next let $s$ be given with $\frac{1}{2}<s<1$. So for $t_\epsilon>0$ small enough, we will have that
\begin{itemize}
\item $\frac{-i}{t_\epsilon}tr(h_{\vartheta}^{1/2} [h_{\vartheta}^{it_\epsilon} h_{\phi}^{-it_\epsilon}-\I]h_{\vartheta}^{1/2})$ is within $\epsilon$ of $S(\vartheta|\phi)$, 
\item and $\frac{-i}{t_\epsilon}h_{\vartheta}^{s} [h_{\vartheta}^{it_\epsilon} h_{\phi}^{-it_\epsilon}-\I]h_\phi^{1-s}$ is within $\epsilon$ of $-i\frac{d}{dt}\left.(h_{\vartheta}^{s}(D{\vartheta}:D{\phi})_t h_{\phi}^{1-s})\right\bracevert_{t=0}=h_{\vartheta}^s \cdot \log h_{\vartheta} \cdot h_{\phi}^{1 - s} - h_{\vartheta}^{s} \cdot \log h_{\phi} \cdot h_{\phi}^{1 - s}$ with respect to $L^1$-norm.
\end{itemize}
Notice that by the properties of the trace functional $tr$ we have that $tr(h_{\vartheta}^{s} [h_{\vartheta}^{it_\epsilon} h_{\phi}^{-it_\epsilon}-\I]h_\phi^{1-s})=tr(h_{\vartheta}^{1/2} [h_{\vartheta}^{it_\epsilon} h_{\phi}^{-it_\epsilon}-\I]h_\phi^{1-s}h_{\vartheta}^{s-1/2})$. 

By assumption $\phi\leq \vartheta(\delta)$. This means that $t\to(D\phi:D\vartheta)_t$ extends to an $\qM$-valued function $f(z)$ which is point to weak*-continuous on the closed strip, $\{z\in\mathbb{C}: -\delta\leq\Im(z)\leq 0\}$, and analytic on the open strip $\{z\in\mathbb{C}: -\delta<\Im(z)< 0\}$. For each $z$ the value $f(z)$ is essentially just an extension of $h_\phi^{iz}h_{\vartheta}^{-iz}$. (For details of this construction see \cite{Kos}). In view of this we will simply write $[h_\phi^{iz}h_{\vartheta}^{-iz}]$ for $f(z)$. So if we set $z=ir$ where $0\leq r\leq \delta$, we obtain that as $r\searrow 0$ we will have that $[h_\phi^{-r}h_{\vartheta}^r]\to \I$ in the weak* topology on $\qM$. 

Next notice that for $0\leq 1-s\leq\delta$, we have that $h_\phi^{1-s}h_{\vartheta}^{s-1/2}= [h_\phi^{1-s}h_{\vartheta}^{-(1-s)}]h_\vartheta^{1/2}$.
Therefore as $s\nearrow 1$ on the interval $[1-\delta, 1]$, we must have that $\frac{-i}{t_\epsilon}tr(h_{\vartheta}^{1/2} [h_{\vartheta}^{it_\epsilon} h_{\phi}^{-it_\epsilon}-\I]h_\phi^{1-s}h_{\vartheta}^{s-1/2})=\frac{-i}{t_\epsilon}tr(h_{\vartheta} [h_{\vartheta}^{it_\epsilon} h_{\phi}^{-it_\epsilon}-\I][h_\phi^{1-s}h_{\vartheta}^{-(1-s)}])$ converges to $\frac{-i}{t_\epsilon}tr(h_{\vartheta}[h_{\vartheta}^{it_\epsilon} h_{\phi}^{-it_\epsilon}-\I])=\frac{-i}{t_\epsilon}tr(h_{\vartheta}^{1/2} [h_{\vartheta}^{it_\epsilon} h_{\phi}^{-it_\epsilon}-\I]h_{\vartheta}^{1/2})$.

There must therefore exist a $\widetilde{\delta}>0$ such that for any $s$ with $1-\widetilde{\delta} <s <1$, the term $\frac{-i}{t_\epsilon}tr(h_{\vartheta}^{1/2} [h_{\vartheta}^{it_\epsilon} h_{\phi}^{-it_\epsilon}-\I]h_\phi^{1-s}h_{\vartheta}^{s-1/2})$ will be within $\epsilon$ of $\frac{-i}{t_\epsilon}tr(h_{\vartheta}^{1/2} [h_{\vartheta}^{it_\epsilon} h_{\phi}^{-it_\epsilon}-\I]h_{\vartheta}^{1/2})$. If we combine all the above observations, it follows that for any $s$ with $1-\widetilde{\delta} <s <1$, $tr(h_{\vartheta}^{s}\cdot \log h_{\vartheta} \cdot  h_{\phi}^{1 - s} - h_{\vartheta}^{s} \cdot \log h_{\phi} \cdot h_{\phi}^{1 - s})$  will be within $3\epsilon$ of $S(\vartheta|\phi)$. This proves the theorem.
\end{proof}

\section{An alternative approach to the general quantum case}\label{sec5} 

Here we propose a means for defining the entropy of a single state $\vartheta$. This definition turns out to be equivalent to von Neumann entropy in the tracial case. Some careful preparation for and justification of this definition is required. As a first step in identifying a suitable prescription for defining entropy of a single state, we take some time to see what Theorem \ref{thm4.3} looks like when the states in question commute. A crucial tool in this endeavour, is the recently developed theory of Orlicz spaces for type III algebras (see \cite{L}). A crucial construct in the development of this theory of ``type III'' Orlicz spaces, is the concept of a fundamental function. The fundamental function of a rearrangement invariant Banach function space on $(\mathbb{R},\mathcal{B}_{\mathbb{R}}, \lambda)$, say $L^\rho(\mathbb{R})$, is defined on $[0,\infty)$ by the prescription $\varphi(t) = \|\chi_E\|_\rho$, where $E$ is any measurable subset of $\mathbb{R}$ with $\lambda(E)=t$. The rearrangement invariance of the space in question, ensures the well-definedness of the corresponding fundamental function. The interested reader may find a more detailed introduction to fundamental functions on pages 65-73 of \cite{BS}. The two facts regarding fundamental functions that we need, is that for an Orlicz space $L^\Psi(\mathbb{R})$, the fundamental function is given by the prescription $t\to\frac{1}{\Psi^{-1}(1/t)}$ when the Luxemburg norm is in view, and by $t\to t(\Psi^*)^{-1}(1/t)$ when the Orlicz norm is in view. (See \cite[II.5.2, IV.8.15 \& IV.8.17]{BS}.) We will need the following lemma in our investigation. (The proof is contained in the proof of \cite[Theorem 2.2]{L}.)

\begin{lemma}\label{mainthm}
Let $a, b\, \eta \,\qM^+$ be commuting affiliated operators. Let $\Psi$ be an Orlicz function and let $\varphi_\psi$ be the fundamental function of $L^\Psi(\mathbb{R})$ equipped with the Luxemburg norm. Then $$\chi_{(1,\infty)}(a\varphi_\Psi(b)) = \chi_{(1,\infty)}(\Psi(a)b).$$
\end{lemma}

\begin{proof}
Let $\alpha, \beta > 0$ be given. It is a known fact that $\alpha \Psi(\beta) \leq 1 \Leftrightarrow \beta \leq \Psi^{-1}(\frac{1}{\alpha})$. If we apply this fact to the Borel functional calculus for the commuting positive operators $a$ and $b$, we have that $\chi_{(1,\infty)}(a\varphi_\Psi(b)) = \chi_{(1,\infty)}(\Psi(a)b)$ as required.
\end{proof}

The above lemma now enables us to make the following conclusion:

\begin{proposition}\label{prop5.2}
Let $\vartheta$, $\phi$ be faithful normal states on $\qM$ with unit vector representatives $h_\vartheta^{1/2}$, $h_\phi^{1/2}$, which commute in the sense that they satisfy one (and therefore all) of the criteria described in \cite[Corollary VIII.3.6]{Tak2}. Assume in 
addition that $\phi\leq \vartheta(\delta)$ for some $\delta>0$, . With $\varphi_{\log}$ denoting the fundamental function of 
the space $L\log(L+1)(\mathbb{R})$ (equipped with the Luxemburg norm), we then have that 
\begin{itemize}
\item $h_\vartheta$ and $h_\phi$ are commuting operators affiliated to $\mathcal{M}$, 
\item $f= h_\vartheta h_\phi^{-1}$ extends uniquely to an element of $\qM$,
\item and $S(\vartheta|\phi)= \phi(f\log(f)) =\inf_{\epsilon >0} [\epsilon\tau(\chi_{(\epsilon, \infty)}(\varphi_{\log}(h_\phi)f) + \log(\epsilon)\|h_\phi f\|_1]$.
\end{itemize}
\end{proposition}

\begin{proof}
The first step is to show that $h_\vartheta$ and $h_\phi$ commute. 

It is clear from the proof of \cite[Corollary VIII.3.6]{Tak2}, that the commutation of  $\vartheta$ and $\phi$, ensures the existence of an operator $h$ affiliated to the centraliser $\qM_\phi$ of $\phi$ for which we have that $(D\phi:D\vartheta)_t=h^{it}$. But from the discussion preceding Theorem \ref{thm4.3} we know that $(D\vartheta:D\phi)_t=(D\widetilde{\vartheta}:D\widetilde{\phi})_t= h^{it}_\vartheta h^{-it}_\phi$. In other words for each $t$, $h^{it} =  h^{it}_\vartheta h^{it}_\phi$.

On appealing to the properties of the cocycle derivative, we may now conclude that
\begin{eqnarray*}
h^{i(t+s)}&=&(D\widetilde{\vartheta}:D\widetilde{\phi})_{t+s}\\
&=&(D\widetilde{\vartheta}:D\widetilde{\phi})_{s}\sigma^\phi_s((D\widetilde{\vartheta}:D\widetilde{\phi})_{t})\\
&=& h^{is}h_\phi^{is} h^{it} h_\phi^{-is}
\end{eqnarray*}
or equivalently, $h^{it}=h_\phi^{is} h^{it} h_\phi^{-is}$. So each $h^{it}$ commutes with each $h_\phi^{is}$.

But we saw earlier that $h^{it} = h^{it}_\vartheta h^{-it}_\phi$, or equivalently that $h^{it}_\vartheta=h^{it}h^{it}_\phi$. Together these two facts ensure that each $h_\vartheta^{it}$ commutes with each $h_\phi^{is}$. We may now use the Borel functional calculus to conclude from these two facts that $h_\phi$ and $h_\vartheta$ themselves also commute. This proves the first bullet.

To see the second bullet, we note from the proof of Theorem \ref{thm4.3} that the requirement that $\phi\leq \vartheta(\delta)$, ensures that for $r>0$ small enough, $h_\vartheta^r h_\phi^{-r}$ extends uniquely to an element of $\qM$. Since $h_\phi$ and $h_\vartheta$ commute, this clearly ensures that the closure of $(h_\vartheta^r h_\phi^{-r})^{1/r}=h_\vartheta h_\phi^{-1}$ also belongs to $\qM$.

For the final bullet, note that by the Borel functional calculus, the commutation of $h_\phi$ and $h_\vartheta$, ensures that we may write the limit formula $\lim_{s\nearrow 1} tr(h_{\vartheta}^{s}\cdot \log h_{\vartheta} \cdot  h_{\phi}^{1 - s} - h_{\vartheta}^{s} \cdot \log h_{\phi} \cdot h_{\phi}^{1 - s})$ as $\lim_{s\nearrow 1} tr(f^s\log(f)h_{\phi})$ where $f=h_\vartheta h_\phi^{-1}$. It also follows from the proof of Theorem \ref{thm4.3}, that there exists an interval $[\delta, 1]$ for which $s\to f^s$ is point-weak* continuous. So given $\rho$ with $\delta<\rho<1$, we may write the limit formula as $\lim_{r\nearrow \rho} tr(f^r[f^{(1-\rho)}\log(f)]h_{\phi})$. We may now use the continuous functional calculus to see that since $f\in \qM$, we must have that $f^{(1-\rho)}\log(f) \in \qM$. But then $[f^{(1-\rho)}\log(f)]h_{\phi}\in L^1(\qM)$. The point-weak* continuity of the map 
$r\to f^r$ on $[\delta, \rho]$, now ensures that $\lim_{r\nearrow \rho} tr(f^r[f^{(1-\rho)}\log(f)]h_{\phi})=tr(f\log(f)h_\phi)=\phi(f\log(f))$.

To prove the final equality, one firstly uses a similar argument to the one in 
\cite[Proposition 6.8]{[2]} to see that $tr(h_\phi f\log(f))= \inf_{\epsilon >0}[\epsilon tr(h_\phi(f/\epsilon)\log((f/\epsilon)+\I)) +\log(\epsilon)tr(h_\phi f)]$. On 
combining the preceding Lemma with \cite[Lemma II.5 \& Def II.13]{terp}, we then have that 
$$tr(h_\phi(f/\epsilon)\log((f/\epsilon)+\I)) = \tau(\chi_{(1,\infty)}(h_\phi(f/\epsilon)\log((f/\epsilon)+\I))$$ 
$$=\tau( \chi_{(1,\infty)} (\varphi_{\log}(h_\phi)f/\epsilon))= \tau(\chi_{(\epsilon,\infty)}(\varphi_{\log}(h_\phi)f)).$$This proves the final claim.
\end{proof}

\bigskip

We are now finally ready to present the definition of the entropy ${S}(\vartheta)$ of a faithful normal state $\vartheta$. The basic idea is to use the above result as guide, for the kind of technical prescription that might work. Tempting as it may be to simply replace $f=h_\vartheta h_\phi^{-1}$ with $h_\vartheta$, and $\phi$ with $tr$, to obtain $tr(h_\vartheta\log(h_\vartheta))$ as a definition, this cannot possibly work. The problem with this prescription is that $tr$ is only defined on $L^1(\qM)$ where in the crossed product setting, all the elements $h$ of $L^1(\qM)$ have to satisfy the requirement that $\theta_s(h)=e^{-s}h$ for all $s$. Since $h_\vartheta\in L^1(\qM)$ we do have that $\theta_s(h_\vartheta)=e^{-s}h_\vartheta$, But then $\theta_s(h_\vartheta\log(h_\vartheta))=e^{-s}h_\vartheta\log(e^{-s}h_\vartheta)\neq e^{-s}h_\vartheta\log(h_\vartheta)$. However the final equality in the third bullet of Proposition \ref{prop5.2}, does present us with a means for overcoming this difficulty for a subspace of $L^1(\qM)$. The subspace in question is the noncommutative Orlicz space $L^1\cap L\log(L+1)(\qM)$.
 
Some analysis is necessary before we are able to present the definition.  Note that classically $L^1\cap L\log(L+1)$ is an Orlicz space produced by the Young's function 
$$\Psi_{ent}(t)=\max(t,t\log(t+1))=\left\{ \begin{array}{ll} t & 0\leq t\leq e-1\\ t\log(t+1) & e-1\leq t\end{array} \right. $$We start by describing how to construct the type III analogue of the space $L^1\cap L\log(L+1)$. We will for simplicity of computation assume that each of $L\log(L+1)(0,\infty)$ and $L^1\cap L\log(L+1)(0,\infty)$ are equipped with the 
Luxemburg norm. It is then an exercise to see that the fundamental function of $L^1\cap L\log(L+1)(0,\infty)$ is of the form $\varphi_{ent}(t)=\max(t,\varphi_{\log}(t))$. It is this fundamental function that one uses to construct the type III analogue of $L^1\cap L\log(L+1)$ in accordance with the prescriptions given in \cite{L, LM}. Let us for the sake of brevity denote this space by $L^{ent}(\mathfrak{M})$. We now show that this space canonically embeds into both $L^1(\qM)$ and $L\log(L+1)(\qM)$.

From the above computations, it is clear that the functions $\zeta_1(t)= \frac{t}{\varphi_{ent}(t)}$, and $\zeta_{\log}(t)= \frac{\varphi_{\log}(t)}{\varphi_{ent}(t)}$ are both continuous and bounded above (by 1) on $(0,\infty)$. Hence for $h=\frac{D\tilde{\omega}}{D\tau}$, the operators $\zeta_1(h)$ and $\zeta_{\log}(h)$ are both contractive 
elements of $\cM$. It is now an exercise to see that the prescriptions $x\to\zeta_1(h)^{1/2}x\zeta_1(h)^{1/2}$ and $x\to\zeta_{\log}(h)^{1/2}x\zeta_{\log}(h)^{1/2}$ respectively yield continuous embeddings of $L^{ent}(\mathfrak{M})$ into 
$L^1(\mathfrak{M})$ and $L\log(L+1)(\mathfrak{M})$. Using these embeddings, we now make the following definition:

\begin{definition}\label{genentropy}
\label{defent}
A state $\vartheta$ on the von Neumann algebra $\qM$ is called \emph{regular} if for some element $g$ of $[L\log(L+1)\cap L^1](\mathfrak{M})^+=L^{ent}(\mathfrak{M})^+$, 
$\frac{D\tilde{\vartheta}}{D\tau}$ is of the form $\zeta_1(h)^{1/2}g\zeta_1(h)^{1/2}$. For such a regular state we then define the entropy to be 
$$\tilde{S}(\vartheta)=\inf_{\epsilon >0}[\epsilon\tau(\chi_{(\epsilon, \infty)}(\zeta_{\log}(h)^{1/2}g\zeta_{\log}(h)^{1/2}))
+ \log(\epsilon)\|\zeta_1(h)^{1/2}g\zeta_1(h)^{1/2}\|_1].$$(Here $h$ is the density $\frac{D\widetilde{\omega}}{D\tau}$ of the dual weight $\widetilde{\omega}$.)
\end{definition}

\begin{remark}
What we must clarify is the meaning of the term ``regularization'' used in the above definition. The density $h = \frac{D\widetilde{\omega}}{D\tau}$ is related, by Bisognano-Wichmann results \cite{BW1}, \cite{BW2}, to the equilibrium hamiltonian, cf Remark 2.11 in \cite{LM}. This gives a relation to thermodynamics of equilibrium states, cf Remark \ref{2.4a}. Further, $\frac{D\widetilde{\omega}}{D\tau}$ is in $L^1(\qM)$ space. This and Definition \ref{defent}
imply that the regularization procedure stems from the prescription leading to the construction of $L^{ent}(\qM)$ space, see Definition 3.4 in \cite{L}.

On the other hand, it is worth noting that the same procedure was used to define $\tau$-measurability of quantum field operators, see \cite{LM}. Consequently, the regularization procedure is based on the selection of such measurable operators which are good candidates for representing states and this selection is compatible with the new formalism of statistical mechanics. We remind that this new formalism is based on the distinguished pair of Orlicz spaces $\langle L^{\cosh - 1}(\qM), L\log(L+1)(\qM)\rangle$, for details see \cite{[2]}, \cite{[3]}.
\end{remark}
One has (cf \cite{LM})

\begin{corollary}
If $\vartheta$ is a regular state, then $\tilde{S}(\vartheta)$ is well defined (although possibly infinite valued).
\end{corollary}

We proceed to prove a result establishing criteria under which a version of Equation \ref{1.6} holds in the present setting. Note that in this result, the faithful KMS state $\omega$ plays the role of the reference measure $\lambda$.

\begin{theorem}
\label{5.6}
Let $\qM$ be a von Neumann algebra in standard form equipped with the faithful normal state $\omega$ and let $\vartheta$ be a faithful regular state with $h_\vartheta$ of the form $\zeta_1(h)^{1/2}g\zeta_1(h)^{1/2}$, where $g\in (L\log(L+1)\cap L^1)_+(\qM)=L^{ent}_+(\qM)$ commutes with $h_\omega=h$. Then $\tilde{S}(\vartheta)=S(\omega|\vartheta)$.
\end{theorem}

Readers are urged to review the discussion immediately following Remark \ref{rem-4.3} before engaging the proof of this theorem.

\begin{proof} We will write $h_\vartheta$ for $\frac{d\widetilde{\vartheta}}{d\tau}$. By assumption $g$ and $h$, and therefore $h_\vartheta$ and $h$, are commuting affiliated operators. Hence so are $p=h_\vartheta h^{-1}$ and $h$. The proof makes extensive use of the kinds of techniques employed in Theorem \ref{ar-ent}. Given $a\in \qM$ and $b\in L^1_+(\qM)$, we will for this reason once again employ the notational device of writing $\langle ab^{1/2}, b^{1/2}\rangle$ for $tr(ba)$.

In this case, let $\lambda\to e_\lambda$ be the spectral resolution of $p$. The fact that both $h$ and $h_\vartheta$ belong to $L^1(\qM)$ ensures
that $\theta_s(p)=\theta_s(h_\vartheta)\theta_s(h^{-1})=h_\vartheta h^{-1}=p$ for every $s\in \mathbb{R}$. This in turn is enough to ensure that $p$ is actually affiliated with the copy $\pi_\omega(\qM)$ of $\qM$ inside the crossed product $\cM$. Therefore, the spectral projections $e_\lambda$ are all elements
of $\qM$. The first stage of the proof is to show that in general $S(\omega|\vartheta)$ is finite if and only if the integral $\int_0^\infty\lambda\log(\lambda)\,d\langle e_\lambda h^{1/2}, h^{1/2}\rangle$ converges, in which case they are equal (where the integral $\int_0^\infty\lambda\log(\lambda)\,d\langle e_\lambda h^{1/2}, h^{1/2}\rangle$ represents $\langle p\log(p)h^{1/2},h^{1/2}\rangle= tr(h^{1/2}p\log(p)h^{1/2})$). Additionally, that the only way the integral can diverge is by diverging to $\infty$. 

Araki entropy is formulated using technology developed for the GNS setting. So one of the main challenges we need to overcome to achieve the stated objective, is to understand how to translate expressions like $\langle p\log(p)h^{1/2}, h^{1/2}\rangle$ to the GNS setting. Here we shall in passing make use of the equivalence of Haagerup $L^p$-spaces with Araki-Masuda and Connes-Hilsum $L^p$-spaces.

By \cite[Theorem II.37 \& Corollary II.38]{terp} we may freely pass to the setting where $\vartheta$ rather than $\omega$ is being used as the reference weight. If after having done this we then use the weight 
$\vartheta'=\vartheta(J\cdot J)$ on $\qM'$ to construct Connes-Hilsum spaces, it follows from a combination of \cite[Theorem 2]{Hil} and \cite[Proposition 11.6]{Hiai-book} that, in moving from the 
Haagerup to the Connes-Hilsum context, $h_\omega$ and $h_\vartheta$ will then respectively be mapped onto $\Delta_{\omega,\vartheta} \otimes L (\equiv \frac{d\omega}{d\vartheta'}\otimes L)$ and 
$\Delta_{\vartheta} \otimes L (\equiv \frac{d\vartheta}{d\vartheta'}\otimes L)$ by means of a spatial *-isomorphism, where $L$ is here a densely defined positive operator on $L^2(\mathbb{R})$. (See 
also the discussion preceding Lemma 11.30 and equation (11.25) of \cite{Hiai-book}.) The commutation of the pair $h_\omega$ and $h_\vartheta$  is of course carried over to the pair 
$\Delta_{\omega,\vartheta}\otimes L$ and $\Delta_\vartheta\otimes L$. In this context the strong product $p= h_\vartheta h_\omega^{-1}$ will then correspond to the strong product of  
$\Delta_{\omega,\vartheta}^{-1}\otimes L^{-1}$ and $\Delta_{\vartheta}\otimes L$, namely $\Delta_\vartheta\Delta_{\omega,\vartheta}^{-1}\otimes \I$. By \cite[Lemma 11.30]{Hiai-book} the spatial 
isomorphism maps $\pi(\qM)$ onto $\qM\otimes\I$ and since the affiliation of $p$ to $\pi_\vartheta(\qM)$ is preserved by the spatial isomorphism, the operator 
$\Delta_\vartheta\Delta_{\omega,\vartheta}^{-1}\otimes \I$ is affiliated to $\qM\otimes\I$. From the description of the action of the spatial isomorphism on $\pi(\qM)$ in \cite[Lemma 11.30]{Hiai-book} 
this can only be if in fact $p=\pi(\Delta_\vartheta\Delta_{\omega,\vartheta}^{-1})$, which in turn ensures that $tr(h^{1/2}p\log(p)h^{1/2}) =\langle p\log(p)h_\omega^{1/2},h_\omega^{1/2}\rangle = 
\langle\Delta_\vartheta\Delta_{\omega,\vartheta}^{-1}\log(\Delta_\vartheta\Delta_{\omega,\vartheta}^{-1}) \Omega, \Omega\rangle$.

With the translation of the inner product $\langle p\log(p)h^{1/2},h^{1/2}\rangle$ from the one context to the other clarified, we pass to the next phase of the proof. We may make sense of the quantity $\langle p\log(p)h^{1/2},h^{1/2}\rangle$ as a spectral integral where $\langle p\log(p)h^{1/2},h^{1/2}\rangle = 
\int_0^\infty\lambda\log(\lambda)\,d\langle e_\lambda h^{1/2}, h^{1/2}\rangle$. Since $x\log(x)$ induces a bounded continuous function on [0,1], this integral diverges iff 
$\int_1^\infty\lambda\log(\lambda)\,d\langle e_\lambda h^{1/2}, h^{1/2}\rangle = \infty$. On now using the fact that $\Delta_{\omega,\vartheta}^{1/2}\Theta=\Omega$ 
(see \cite[Proposition 10.6]{Hiai-book} for this) and appealing to the Borel functional calculus for commuting densely defined operators, it now follows that 
\begin{eqnarray*}
\langle p\log(p)h^{1/2},h^{1/2}\rangle &=& \langle\Delta_\vartheta\Delta_{\omega,\vartheta}^{-1}\log(\Delta_\vartheta\Delta_{\omega,\vartheta}^{-1}) \Omega, \Omega\rangle\\
&=& \langle\Delta_\vartheta\Delta_{\omega,\vartheta}^{-1}\log(\Delta_\vartheta\Delta_{\omega,\vartheta}^{-1})\Delta_{\omega,\vartheta}^{1/2}\Theta, \Delta_{\omega,\vartheta}^{1/2}\Theta\rangle\\
&=& \langle\Delta_\vartheta(\log(\Delta_\vartheta)-\log(\Delta_{\omega,\vartheta}))\Theta, \Theta\rangle\\
&=& -\langle\log(\Delta_{\omega,\vartheta})\Theta, \Theta\rangle\\
&=& S(\omega|\vartheta).
\end{eqnarray*}
Here we used the fact that $\Theta$ is an eigenvector of $\Delta_\vartheta$ with eigenvalue 1 to obtain the second to last equality.

The next part of the proof is to show that (infinite values included), we have that $$\inf_{\epsilon >0}\int_0^\infty\lambda\log(\lambda+\epsilon)\,d\langle e_\lambda h^{1/2}, h^{1/2}\rangle= \int_0^\infty\lambda\log(\lambda)\,d\langle e_\lambda h^{1/2}, h^{1/2}\rangle$$from which we will then be able to deduce the claim. To see this next fact, observe that for any $\lambda>0$ and $\epsilon>0$, we have that $\log(\lambda+\epsilon)>\log(\lambda)$ and that $$0< \lambda(\log(\lambda+\epsilon)-\log(\lambda))=\lambda\log(1+(\epsilon/\lambda))\leq\epsilon.$$This ensures that
\begin{eqnarray*}
\int_0^\infty\lambda\log(\lambda)\,d\langle e_\lambda h^{1/2}, h^{1/2}\rangle &\leq& \int_0^\infty\lambda\log(\lambda+\epsilon)\,d\langle e_\lambda h^{1/2}, h^{1/2}\rangle\\
&\leq& \int_0^\infty(\lambda\log(\lambda)+\epsilon)\,d\langle e_\lambda h^{1/2}, h^{1/2}\rangle\\
&=& \int_0^\infty\lambda\log(\lambda)\,d\langle e_\lambda h^{1/2}, h^{1/2}\rangle +\epsilon\omega(\I)\\
&=& \int_0^\infty\lambda\log(\lambda)\,d\langle e_\lambda h^{1/2}, h^{1/2}\rangle +\epsilon,
\end{eqnarray*}
which establishes the claim.

If we combine the two facts we have proved thus far, it yields the conclusion that (infinite values included) we always have that
$$S(\omega|\vartheta)=\inf_{\epsilon >0}\int_0^\infty\lambda\log(\lambda+\epsilon)\,d\langle e_\lambda h^{1/2}, h^{1/2}\rangle.$$We now use this formula to show that $S(\omega|\vartheta)=\tilde{S}(\vartheta)$. This claim will follow if we can show that for any
$\epsilon>0$, the equality $\epsilon\tau(\chi_{(\epsilon, \infty)}(\zeta_{\log}(h)^{1/2}g\zeta_{\log}(h)^{1/2}))
+ \log(\epsilon)\|\zeta_1(h)^{1/2}g\zeta_1(h)^{1/2}\|_1= \int_0^\infty(\lambda\log(\lambda)+\epsilon)\,d\langle e_\lambda h^{1/2}, h^{1/2}\rangle$ holds true. Let $\epsilon>0$ be given. Firstly, observe that by the definition of $\zeta_1$ and $\zeta_{\log}$, we have $\zeta_{\log}(h)=\varphi_{\log}(h)\varphi_{ent}(h)^{-1}=\varphi_{\log}(h)h^{-1}\zeta_1(h)$. Thus, by the commutation assumption, we
have that $$\zeta_{\log}(h)^{1/2}g\zeta_{\log}(h)^{1/2} = \varphi_{\log}(h)h^{-1}\zeta_1(h)^{1/2}g\zeta_1(h)^{1/2} = \varphi_{\log}(h)h^{-1}h_\vartheta=\varphi_{\log}(h)p.$$Hence we may apply 
Lemma \ref{mainthm} to see that
\begin{eqnarray*}
&&\epsilon\tau(\chi_{(\epsilon, \infty)}(\zeta_{\log}(h)^{1/2}g\zeta_{\log}(h)^{1/2}))
+ \log(\epsilon)\|\zeta_1(h)^{1/2}g\zeta_1(h)^{1/2}\|_1\\
&=& \epsilon\tau(\chi_{(\epsilon, \infty)}(\varphi_{\log}(h)p))
+ \log(\epsilon)tr(hp)\\
&=& \epsilon\tau(\chi_{(1, \infty)}(\varphi_{\log}(h)(p/\epsilon))
+ \log(\epsilon)tr(hp)\\
&=& \epsilon\tau (\chi_{(1, \infty )}(h(p/\epsilon )\log ((p/\epsilon )+\I )))
+ \log(\epsilon)tr(hp).
\end{eqnarray*}
Since $h\in L^1(\qM)$ with $p$ affiliated to $\qM$, the operator $(p/\epsilon)(\log (p/\epsilon)+\I)$ will by the Borel functional calculus itself be affiliated with $\qM$, in which case $b=h(p/\epsilon)(\log (p/\epsilon)+\I)$ will then be a positive operator ``affiliated'' to $L^1(\qM)$ in the sense that it is affiliated with the crossed product and that $\theta_s(b)=e^{-s}b$ for each $s\in \mathbb{R}$. By \cite[Proposition II.4]{terp}, $b$ corresponds to a normal weight $\Phi_b$ on $\qM$. If we now apply \cite[Lemma II.5]{terp}, it follows that $\Phi_b(\I)= \tau (\chi_{(1, \infty)}(h(p/\epsilon)\log ((p/\epsilon)+\I)))$. Writing $e_N$ for $\chi_{[0,N]}(p)$, we next again appeal to \cite[Proposition II.4]{terp} to see that for each $N>0$, the weight $f\to\Phi_b(e_Nfe_N)$ corresponds to $e_Nbe_N$. All of these observations may now be combined with the normality of $\Phi_b$ and the definition of the tracial functional $tr$ (\cite[Definition II.13]{terp}), to see that
\begin{eqnarray*}
&&\epsilon\tau (\chi_{(1, \infty)}(h(p/\epsilon)\log ((p/\epsilon)+\I)))
+ \log(\epsilon)tr(hp)\\
&=& \epsilon\Phi_b(\I)+ \log(\epsilon)tr(hp)\\
&=& \epsilon\lim_{N\to\infty}\Phi_b(e_N)+ \log(\epsilon)tr(hp)\\
&=& \epsilon\lim_{N\to\infty}\tau (\chi_{(1, \infty)}(h(e_Np/\epsilon)\log((e_Np/\epsilon)+\I)))
+ \log(\epsilon)tr(hp)\\
&=& \epsilon\lim_{N\to\infty}tr(h(e_Np/\epsilon)\log((e_Np/\epsilon)+\I))
+ \log(\epsilon)tr(hp)\\
&=& \epsilon\lim_{N\to\infty}\langle (e_Np/\epsilon)\log((e_Np/\epsilon)+\I)h^{1/2}, h^{1/2}\rangle
+ \log(\epsilon)\langle ph^{1/2}, h^{1/2}\rangle\\
&=& \lim_{N\to\infty}\int_0^N \lambda\log((\lambda/\epsilon)+1)\,d\langle e_\lambda h^{1/2}, h^{1/2}\rangle
+ \log(\epsilon)\int_0^\infty\lambda\,d\langle e_\lambda h^{1/2}, h^{1/2}\rangle\\
&=& \int_0^\infty \lambda\log((\lambda/\epsilon)+1)\,d\langle e_\lambda h^{1/2}, h^{1/2}\rangle
+ \log(\epsilon)\int_0^\infty\lambda\,d\langle e_\lambda h^{1/2}, h^{1/2}\rangle\\
&=& \int_0^\infty \lambda\log(\lambda+\epsilon)\,d\langle e_\lambda h^{1/2}, h^{1/2}\rangle.
\end{eqnarray*}
This proves the claim required to establish the theorem. In obtaining the final equality, we silently used the facts that 
$\int_0^\infty \lambda\log((\lambda/\epsilon)+1)\,d\langle e_\lambda h^{1/2}, h^{1/2}\rangle$ either converges or diverges to $\infty$ and that we always have \newline 
$\int_0^\infty\lambda\,d\langle e_\lambda h^{1/2}, h^{1/2} \rangle =tr(hp)=tr(h_\vartheta)=\vartheta(\I)=1<\infty$.
\end{proof}

\begin{remark}
The full significance of Theorem \ref{5.6} will be discussed in Section 6. For now the important point to note here is that to define entropy for large systems (so for type III von Neumann algebras) we were here working within the new formalism, which is based on the distinguished pair of Orlicz spaces $\left\langle L^{\cosh -1}, L\log(L+1)\right\rangle$ - for details see \cite{[1]}, \cite{[2]}, \cite{[3]}. In particular,  the superalgebra $\cM$ was employed. In that way it is possible to define entropy for non-semifinite von Neumann algebras, and consequently to study thermodynamics for such systems. Furthermore, this should make clear in which way we avoided the problems discussed in \cite{OP} - see Theorem 6.10 of that monograph.

Now let $\qM$ be a semifinite von Neumann algebra and $\omega=\tau_\omega$ a tracial state. Let $\vartheta$ be a faithful normal state for which the Radon-Nikodym derivative $a$ described in Theorem \ref{2.1} belongs to the tracial space $[L\log(L+1)\cap L^1](\mathfrak{M},\tau_\omega)$ (see the prescription in for example section 1 of \cite{L} to see how this space is defined. When passing to the crossed product, it is known that in the case of semifinite algebras equipped with a trace (as is the case here), the crossed product $\mathcal{M}$ of $\qM$ with the modular automorphism group of $\tau_\omega$, is essentially just a copy of $\qM\otimes L^\infty(\mathbb{R})$ \cite[Part II, Proposition 4.2]{vD}. In particular, under this correspondence the canonical trace $\tau$ on $\cM$, may be identified with $\tau_\omega \otimes \int_\mathbb{R}\cdot e^{-t}\,dt$ (see section 2 of \cite{L}). This identification forms the background for the analysis in section 2 of \cite{L}, where certain quantities described by the pair $(\cM,\tau)$, may alternatively be described by the pair $(\qM,\tau_\omega)$. By Proposition 2.5 and Definition 3.4 of \cite{L}, $a$ corresponds to an element $g$ of $[L\log(L+1)\cap L^1](\mathfrak{M})^+=L^{ent}(\qM)^+$, which is of the form $g=a\otimes\varphi_{ent}(e^t)$. Again by \cite[Proposition 2.5]{L}, the 
operators $\zeta_{\log}(h)^{1/2}g\zeta_{\log}(h)^{1/2}$ and $\zeta_1(h)^{1/2}g\zeta_1(h)^{1/2}$ are respectively of the form 
$a\otimes\varphi_{\log}(e^t)$ and $a \otimes e^t$. Using the fact that $\tau_\omega \otimes \int_\mathbb{R}\cdot e^{-t}\,dt$, we may therefore apply \cite[Proposition 1.7]{FK} and \cite[Theorem 2.2]{L} to see that we will for any $\epsilon>0$ have that 
\begin{eqnarray*}
&&\epsilon\tau ( \chi_{(\epsilon, \infty )}(\zeta_{\log}(h)^{1/2}g\zeta_{\log}(h)^{1/2}))
+ \log(\epsilon)\|\zeta_1(h)^{1/2}g\zeta_1(h)^{1/2}\|_1\\
&&\qquad =\epsilon\tau(\chi_{(\epsilon, \infty)}(\zeta_{\log}(h)^{1/2}g\zeta_{\log}(h)^{1/2}))+ \log(\epsilon)\tau ( \chi_{(1, \infty )}(\zeta_1(h)^{1/2}g\zeta_1(h)^{1/2}))\\
&&\qquad =\tau_\omega(a\log (a/\epsilon +\I))+ \log(\epsilon)\tau_\omega(a)\\
&&\qquad =\tau_\omega (a\log (a +\epsilon\I ))
\end{eqnarray*}
So in this case the formula in the preceding definition corresponds exactly to the more familiar formula $\tilde{S}(\vartheta)=\inf_{\epsilon>0}\tau_\omega(a\log(a+\epsilon\I))= \tau_\omega(a\log(a))$.
\end{remark}

\section{Discussion}
As was noted in the introduction, the standard framework of classical statistical mechanics is based on the pair
\begin{equation}
\label{1d}
\langle L^{\infty}(\Gamma, \mu), L^1(\Gamma, \mu)\rangle,
\end{equation}

for a measure space $(\Gamma, \mu)$. Let us consider this point in detail. There are two ``extremal cases'' of a measure spaces which are employed in Physics. The first case is a countably totally atomic measure space while the second one is based on non-atomic measure.

Let us consider the first case. Then (\ref{1d}) reads
\begin{equation}
\label{2d}
\langle l^{\infty}(\Nn), l^1(\Nn)\rangle
\end{equation}
and then the states are described by
\begin{equation}
\label{3d}
\{ f \equiv (f_1,f_2,...) \in l^1; f\geq 0, \sum_i f_i = 1 \} \subset l^1.
\end{equation}
It is important to note that in (\ref{3d}) one has pure states and a general state is a convex combination of pure states. Furthermore, if in $l^1$ there are only finite sequences, so when $l^1 \equiv l^1(1,2,...,N)$, then Boltzmann's $W$-entropy follows from the recipe for the $H$-functional, provided that the probability distribution is uniform.

On the other hand, the second case leads to
\begin{equation}
\label{4d}
\langle L^{\infty}(\Gamma, d\mu), L^1(\Gamma, d\mu)\rangle
\end{equation}
with the states then given by
\begin{equation}
\label{5d}
\{ f \in L^1(\Gamma, d\mu); f\geq 0, \int f d\mu = 1 \}.
\end{equation}
where the reference measure $\mu$ is non-atomic.
It is crucial to note that in (\ref{5d}) there do not exist pure normal states. Therefore, if as in Boltzmann's theory, the reference measure is akin to Lebesgue measure in the sense of being non-atomic, an examination of the behaviour of the H-functional with respect to pure normal states is an example of ``ill posed'' problem.

Turning to quantization, the von Neumann entropy (based on Dirac's formalism) uses pure states and hence is related to (\ref{3d}).
Contrariwise, a general quantum system, cf Sections 4 and 5, needs to allow for type III von Neumann algebras. It is known that a type III factor $\mathfrak{M}$ does not have normal pure states. Therefore,  type III factors have that mathematical feature in common with the abelian von Neumann algebra $L^{\infty}(\Gamma, d\mu)$ given in (\ref{4d}) which also has no pure normal states.

Consequently the entropy $\widetilde{S}(\vartheta)$ defined in Definition \ref{defent} in the previous section, has more in common with the H-functional, than with the von Neumann entropy. Before proceeding further let us pause to make some important remarks on the nature of states.
\begin{remark}
\begin{itemize}
\item Although in (\ref{5d}) there are no pure normal states, any probability measure is an accumulation point of the convex hull of Dirac measures. This property of classical measure theory (the weak-$^*$ Riemann approximation property) implies that for a continuous classical system all states are separable, see \cite{WAM}. Furthermore, interpreting a Dirac measure as a pure state, one can again say that a convex combination of pure states leads to a state.
\item A non-commutative integral, does not in general have the weak-$^*$ Riemann approximation property. Thus, there is a ``room'' for entangled states, see \cite{WAM}.
\item As a type III von Neumann algebra has no pure normal states, the question of whether $\widetilde{S}(\vartheta)$ is zero only for pure states has no sense.
\item Finally, to avoid any confusion, we note that each $W^*$ algebra is also a $C^*$ algebra with unit. So, the set of all states of such algebra has pure states (by the Krein-Milmann theorem) but these states are not normal!
\end{itemize}
\end{remark}

Turning to the H-functional, we note that it is an easy observation that $H({\chi}_{{}_{\Gamma_{0}}}) = 0$, where ${\chi}_{{}_{\Gamma}}$ is a characteristic function given by a measurable subset $\Gamma_0 \subset \Gamma$.
Clearly, $\chi_{{}_{\Gamma_{0}}}$ is a projector in $L^{\infty}(\Gamma, d\mu)$. However, we are again not able to simplistically translate this property of the H-functional to general quantum systems. To see this, let us assume that a projector $P$ is in 
$L^1(\qM)$. This means that $\theta_s(P) = e^{-s}P$, for any $s$, where $\theta_s$ stands for the dual action of $\Rn$ on $\cM$. But, one has also
$$\theta_s(P) = \theta_s (P \cdot P) = \theta_s(P) \theta_s(P) = e^{-2s}P,$$
which is only possible for $s=0$. The problem here is that the entropy defined in the previous section only makes sense for elements of $L^{ent}(\qM)$. So to make sense of the ``entropy'' of a projector $P$, we first have to embed $P$ into the space $L^{ent}(\qM)$. If indeed $\omega(P)<\infty$ (where $\omega$ is the a priori given faithful normal semifinite weight on $\qM$), then $g=\varphi^{1/2}_{ent}(h)(P)\varphi^{1/2}_{ent}(h)$ (where $h=\frac {d\tilde{\omega}}{d\tau}$) belongs to $L^{ent}(\qM)$ whenever $\omega(P)<\infty$ \cite[Proposition 3.3]{L}. The quantum analogue of   $H({\chi}_{{}_{\Gamma_{0}}})$ would then be given by applying the prescription in Definition \ref{defent} with $g$ as above. Hence, to sum up:
\vskip 0.5cm
\textit{
Some basic properties of classical entropy S as well as of the H-functional have no quantum counterparts in the theory based on type III von Neumann algebras. In particular, the entropy $\widetilde{S}(\vartheta)$ does not exhibit some of the properties typical of its classical counterparts. This is not surprising as entropy being a function of states, should at some level reflect the structure of the state space of the considered system.}

\vskip 0.5cm
However, despite the above differences between the classical and quantum descriptions, the new approach presented here offers a solution to old open problems.  It is well known that in classical statistical mechanics, the Gibbs Ansatz $Z^{-1}e^{-\beta H}$, is designed to describe a classical canonical equilibrium state and that essential thermodynamical information is contained in the partition function $Z = \int e^{-\beta H}d\Gamma$. Here $H$ stands for the Hamiltonian of the considered system, and $\beta$ for the ``inverse'' temperature. The quantization of $e^{-\beta H}$ means that now $H$ is the Hamiltonian operator, and hence to have a quantum state within Dirac's formalism, we require that $e^{-\beta H}$ should then be a trace class operator. But this is only the case when, at the very least, $H$ has a pure point spectrum with accumulation point at infinity. Unfortunately, even the Hamiltonian of the Hydrogen atom does not fulfill this requirement.

To see that this question has an easy solution in the presented framework we note:
\begin{enumerate}
\item As we have seen in Section 4, there is $h_{\omega} = \frac {d\tilde{\omega}}{d\tau}$ where we are using the ``language'' of non-commutative integration theory, cf the previous sections and/or see \cite{terp}.
\item $h_{\omega}^{it}$ can be identified with $\lambda(t)$.
\item $\theta_s(\lambda(t)) = e^{-ist} \lambda(t)$
\item Writing $\lambda(t) = e^{-iHt}$ one has: 
$$ \theta_s(e^{-iHt}) = e^{-ist} e^{-iHt} = e^{-i(H+s\jed)t}$$
\item Thus $\theta_s(H) = H+ s\jed$
\item Consequently  $\theta_s(e^{-H}) = e^{-s} e^{-H}$ and $e^{-H} \in L^1(\qM)$!
\item In the above $\beta = 1$, which follows from the standard scaling of temperature in the KMS theory, cf Chapter 5 in \cite{BR}.
\end{enumerate}
 Consequently, \textit{the quantum analogue of the Gibbs Ansatz is well defined as an element of $L^1(\qM)$}. Furthermore, as there is a linear bijective isometry between $L^1(\qM)$ and $\qM_*$, we obtain a well defined normal functional on $\qM$. In particular, \textit{the quantum analogue of the partition function is also well defined}. 

Turning to the entropy $\widetilde{S}(\vartheta)$, we wish to get a better understanding of its nature. To this end we will consider the important case when $\vartheta$ is the reference state $\omega$. We remind that $\omega$ is a faithful normal state and by  Takesaki's theorem, cf Section 2, $\omega$ is a KMS state in terms of the modular dynamics. In other words, $\omega$ describes the given equilibrium state and we wish to compute the entropy of such a state. Furthermore, $\widetilde{S}(\omega)$ being related to the equilibrium state $\omega$ is a candidate for quantum thermodynamical entropy.

We have just seen that $e^{-H}$ is in $L^1(\qM)$. But $\widetilde{S}(\vartheta)$ was defined for functionals of the form $\zeta_1(h)g\zeta_1(h)$, where $h \equiv e^{-H}$, and $g \in L^{ent}(\qM)$. So we must examine what this requirement would mean for $e^{-H}$.

We remind, cf. Definition \ref{defent}, that a state ${\vartheta}$ is regular if $\frac{D\tilde{\vartheta}}{D\tau}$ is of the form $\zeta_1(h)g\zeta_1(h)$, where $h \equiv e^{-H}$, and $g \in L^{ent}(\qM)$. Hence we wish to have 
\begin{equation}
e^{-H} \equiv h = \left(\frac{h}{\varphi_{ent}(h)}\right)^{\frac{1}{2}} g \left(\frac{h}{\varphi_{ent}(h)}\right)^{\frac{1}{2}}.
\end{equation}
Thus $g = \varphi_{ent}(h)$, and hence
\begin{equation}
\widetilde{S}(\omega) = \inf_{\epsilon >0}\left[\epsilon\, \tau\left(\chi_{(\epsilon, \infty)}(\zeta_{\log}(h))^{\frac{1}{2}} \varphi_{ent}(h) (\zeta_{\log}(h))^{\frac{1}{2}}\right) + \log \epsilon\, ||\zeta_1(h)^{\frac{1}{2}} \varphi_{ent}(h)\zeta_1(h)^{\frac{1}{2}}||_1 \right]
\end{equation}

We note that (see \cite{FK}),
\begin{equation}
\tau\left(\chi_{(\epsilon, \infty)}(|T|)\right) = \epsilon^{-1} ||T||_1,
\end{equation}
for $T \in L^1(\qM)$. Thus
\begin{eqnarray*}
\widetilde{S}(\omega) &=& \inf_{\epsilon >0}\left[\epsilon\, \tau\left(\chi_{(\epsilon, \infty)}(\zeta_{\log}(h))^{\frac{1}{2}} \varphi_{ent}(h) (\zeta_{\log}(h))^{\frac{1}{2}}\right) + \epsilon \log \epsilon \,\tau\left(\chi_{(\epsilon, \infty)}(\zeta_1(h)^{\frac{1}{2}} \varphi_{ent}(h)\zeta_1(h)^{\frac{1}{2}})\right) \right]\\
&=& \inf_{\epsilon >0}\left[\epsilon\,\tau\left(\chi_{(\epsilon, \infty)}(\varphi_{\log}(h))\right) + \epsilon \log \epsilon\, \tau\left(\chi_{(\epsilon, \infty)}(h)\right) \right]
\end{eqnarray*}
Now observe that $\varphi_{\log}$ is a continuous strictly increasing function which is 0 at 0. So $t\geq \epsilon>0$ if and only if $\varphi_{\log}(t)\geq \varphi_{\log}(\epsilon)>0$, with $\varphi_{\log}(t)\to\infty$ as $t\to\infty$. If we combine this fact with the Borel functional calculus, it is clear that $\chi_{(\epsilon, \infty)}(\varphi_{\log}(h))=\chi_{(\varphi^{-1}_{\log}(\epsilon), \infty)}(h)$. Consequently 
\begin{eqnarray}
\widetilde{S}(\omega) &=& \inf_{\epsilon >0}\left[\epsilon \,\tau\left(\chi_{(\varphi^{-1}_{\log}(\epsilon), \infty)}(h)\right) + \epsilon \log \epsilon \,\tau\left(\chi_{(\epsilon, \infty)}(h)\right) \right]\nonumber\\
&=& \inf_{\epsilon >0}\left[\frac{\epsilon}{\varphi^{-1}_{\log}(\epsilon)} + \log(\epsilon)\right].\|h\|_1
\end{eqnarray}
As $\varphi^{-1}_{\log}(t) = \frac{1}{\Psi^{-1}_{\log}(\frac{1}{t})}$ where $\Psi_{\log}(t)=t\log(t+1)$, we have
\begin{equation}
\varphi^{-1}_{\log}(t) = \frac{1}{\Psi_{\log}(\frac{1}{t})} = \frac{t}{\log(\frac{1}{t} + 1)}.
\end{equation}
So
\begin{equation}
\widetilde{S}(\omega) = \inf_{\epsilon >0} \left[\log\left(\frac{1}{\epsilon} + 1\right) + \log \epsilon\right] ||h||_1 = \left[\inf_{\epsilon >0} \log(1 + \epsilon)\right] ||h||_1 = 0.
\end{equation}

In commenting on this result, we note that in classical Physics, the entropy is an extensive thermodynamical quantity.
The central question then becomes: whether the quantum entropy $\widetilde{S}(\vartheta)$ has the same property. To answer this question we begin by taking closer look at techniques used in definition of $\widetilde{S}(\vartheta)$. The first observation is that, from the very beginning, we employed the approach relevant to a description of large systems, i.e. those systems of statistical physics which can be obtained by thermodynamical limit. The next observation is that Tomita-Takesaki theory was the basic ingredient of our analysis.

It is crucial to note that in the representation induced by a KMS state, basic relations of Tomita-Takesaki theory for finite volume systems survive the thermodynamical limit. In particular, the equilibrium state vector is an eigenvector of $h$ corresponding to eigenvalue $1$ -- for more details see Sections V.1.4 and V.2.3 in \cite{haag}. Furthermore, we have already noted, cf. remark given prior to Theorem \ref{5.6}, that in the presented approach, the state $\omega$ (so a quantum counterpart of probability measure) was used as  a reference measure. On the other hand, in classical statistical physics, the entropy per unit volume is given by $\frac{S(\varrho_{\Lambda})}{V(\Lambda)}$, where $V(\Lambda)$ stands for the volume of the region $\Lambda$. Note that $V(\Lambda)$ is taken with respect to the reference measure (in classical statistical physics, it is the Lebesgue measure). However, having a probability measure as the reference measure one gets $V(\Lambda) = 1$. In other words, $\widetilde{S}(\vartheta_{\Lambda})$ can be considered as the entropy per unit volume. Consequently, the definition of entropy proposed in the paper in together with the regularization procedure, incorporates some basic ideas of thermodynamic limits. 

Thus, the entropy $\widetilde{S}(\vartheta)$, defined in terms suitable for large systems, should share its properties with the density of entropy.
All of this points to the fact that $\widetilde{S}(\vartheta)$ can be considered as an intensive quantity.

To get some intuition about density entropy properties, it seems to be useful to note that the density of entropy for quantum lattice systems is taking its values in the interval $[0, N < \infty]$, where $N$ is the dimension of Hilbert space associated with each site of quantum spin system -- see Section 6.2.4 in \cite{BR}. Finally, the important point to note here is that the result $\widetilde{S}(\omega) = 0$, is compatible with the interpretation of the relative entropy as a ``measure'' of distance between two states, cf Theorem \ref{5.6}. To sum up, we can say that the obtained result $\widetilde{S}(\omega) = 0$ is expected.

With a suitable concept of entropy for regular states of general quantum systems thus having been identified, the challenge now is to develop computational algorithms for this entropy.

\vskip 0.5cm

\section{Conclusions}
One of the challenges of contemporary physics is to derive the macroscopic properties of matter from the quantum laws governing the microscopic description of a system. On the other hand, thermodynamics being a prerequisite for (quantum) statistical physics, provides laws governing the behaviour of macroscopic variables. It is well known that entropy is a crucial concept for this scheme.

Knowing that statistical physics deals with large systems (so systems with infinite degrees of freedom) we proposed a concise approach to entropy. It was done in operator algebraic language. This language is indispensable as on the one hand it is the basis for noncommutative integration theory, and on the other von Neumann algebra of type III are acknowledged to be the correct formalism for large quantum systems. Consequently, using the algebraic approach, a consistent dynamical description of entropy was achieved. 

It is worth pointing out that our results can be considered as the first step in getting genuine quantum thermodynamics for general quantum systems.


\begin{thebibliography}{99}
\bibitem{ara1} H. Araki, Relative entropy of states of von Neumann algebras, \textit{ Publ. RIMS} \textbf{11} (1976) pp. 809-833

\bibitem{ara2} H. Araki, Relative entropy of states of von Neumann algebras II, \textit{ Publ. RIMS} \textbf{13} (1977) pp. 173-192

\bibitem{Araki1} H. Araki, T. Masuda,  Positive cones and $L_p$-spaces for von Neumann algebras, \textit{ Pub. of RIMS} \textbf{18} (1982), pp. 339-411

\bibitem{BS} C. Bennett, R. Sharpley, \textit{Interpolation of operators}, Academic Press,INC. 1988

\bibitem{Ber} S. K. Berberian, \textit{Measure and Integration}, AMS Chelsea Publications (reprint edition), 2011.

\bibitem{BW1} J. J. Bisognano, E. H. Wichmann, On the duality condition for Hermitean scalar fields, \textit{J. Math. Phys.} \textbf{16}, 985 (1975)

\bibitem{BW2} J. J. Bisognano, E. H. Wichmann, On the duality condition for quantum fields, \textit{J. Math. Phys.} \textbf{17}, 303 (1976)

\bibitem{BR} O. Bratteli, D. Robinson, \textit{Operator algebras and Quantum Statistical Mechanics }, Texts and Monographs in Physics, Springer Verlag; vol. I, 1979; vol. II, 1981

\bibitem{Con1} A Connes, Une classification de facteurs de type III, \textit{Ann. Sci. \'Ecole Norm. Sup.,} 4-eme, Ser. 6 (1973), pp. 133-252

\bibitem{Con2} A Connes, Sur le th\'eoreme de Radon-Nikodym pour les normaux fidelis semifini, \textit{Bull. Sci.Math. SecII} \textbf{97} (1973) pp. 253-258

\bibitem{ConTak} A. Connes, M. Takesaki, The flow of weights on factors III, \textit{T\^ohoku Math. Journ.} \textbf{29} (1977) pp. 473-575

\bibitem{connes} A. Connes, {\it Noncommutative geometry}, Academic Press, 1994

\bibitem{FK} T Fack and H Kosaki, Generalized s-numbers of 
$\tau$-measurable operators, \textit{Pacific J Math} \textbf{123}(1986), 
269-300.

\bibitem{haag} R. Haag, \textit{Local Quantum Physics; Fields, Particles, Algebras}, Second Edition, Springer, 1996

\bibitem{haare} U. Haagerup, Operator valued weights in von Neumann algebras, \textit{J. Funct. Anal.} \textbf{32} (1979) pp. 175- 206.

\bibitem{Hiai-book} F. Hiai, \textit{Lectures on Selected Topics in von Neumann Algebras}, EMS Series
of Lectures in Mathematics, vol. 32, European Mathematical Society (EMS), 2021.

\bibitem{Hil} Michel Hilsum, \textit{Les espaces $L^p$ d'une alg\`ebre de von Neumann d\'efinies par
la deriv\'ee spatiale}, J. Funct. Anal. \textbf{40}(2)(1981), 151--169.

\bibitem{chin} A. I. Khinchin, \textit{Mathematical foundations of Statistical Mechanics}, Dover Publications Inc., 1949

\bibitem {Kos} H Kosaki, A remark on Sakai's quadratic Radon-Nikodym theorem, \textit{Proc AMS} \textbf{116}(1992), 783--786.

\bibitem{L} L. E. Labuschagne, A crossed product approach to Orlicz spaces, \textit{Proc. London Math. Soc.} \textbf{107} (2013) 965-1003

\bibitem{[1]} L. E. Labuschagne, W. A. Majewski, Maps on noncommutative Orlicz spaces, \textit{Illinois J. Math.} 55, 1053-1081 (2011)

\bibitem{LM} L. E. Labuschagne, W. A. Majewski, Integral and differential structures for quantum field  theory, arXiv: 1702.00665v2 [math-ph]

\bibitem{[2]} W. A. Majewski, L. E. Labuschagne, On applications of Orlicz spaces to Statistical Physics,  \textit{Ann. H. Poincar\'e}, 15, 1197-1221 (2014)

\bibitem{[3]} W. A. Majewski, L. E. Labuschagne, Why are Orlicz spaces useful for Statistical Physics? in \textit{Noncommutative Analysis, Operator Theory and Applications;} Eds. D. Alpay et al. Birkhauser-Basel, Series: Linear Operators and Linear Systems, vol 252, 271-283 (2016)

\bibitem{Maj} W. A. Majewski, On quantum statistical mechanics; A study guide. {\it Adv. Math. Phys.} 2017, ID 934717  ;arXiv 1608.06766v2 [math-ph]

\bibitem{WAM} W. A. Majewski, Quantum correlations; quantum probability approach. arXiv 1407.4754v4 [quan-ph]

\bibitem{OP} M. Ohya, D. Petz, \textit{Quantum entropy and its use}, Texts and Monographs in Physics, Springer Verlag, 1993

\bibitem{Peder} G. K. Pedersen, \textit{$C^*$-algebras and their automorphism groups}, Academic Press, 1979

\bibitem{PW} W. Pusz, S. L. Woronowicz, Passive states and KMS states for general quantum systems, \textit{Commun. math. Phys.}, \textbf{58}, 273-290 (1978)

\bibitem{ruelle} D. Ruelle, \textit{Statistical Mechanics; Rigorous Results}, W. A. Benjamin, INC. New York Amsterdam 1969

\bibitem{Tak} M. Takesaki, \textit{Theory of Operator Algebras}, vol. I, Springer, 1979

\bibitem{Tak2} M. Takesaki, \textit{Theory of Operator Algebras}, vol. II, Springer, 2003

\bibitem{terp} M. Terp, {\it $L^p$ spaces associated with von Neumann algebras}, Rapport No 3a, University of Copenhagen, (1981)

\bibitem{thompson} C. J. Thompson, \textit{Mathematical Statistical Mechanics}, The Macmillan Company, New York, 1972.

\bibitem{vD} A. Van Daele, \textit{ Continuous crossed products and type III von Neumann algebras},
Cambridge University Press, 1978

\bibitem{Wit} E. Witten, Notes on some entanglement properties of quantum field theory, arXiv:1803.04993v3 [hep-th]

\end{thebibliography}
\end{document}